\documentclass[draft]{article}
\usepackage[a4paper,margin=25mm]{geometry}
\usepackage{amsmath,amssymb,tikz}
\pdfoutput=1
\parskip\smallskipamount

\let\set\mathbb
\def\<#1>{\langle#1\rangle}

\usepackage{amsthm,amsmath}
\usepackage{tikz}
\usepackage{xfrac}
\usepackage{cite}

\usepackage{mathtools}
\usepackage{enumerate}

\usepackage{mathtools}
\usepackage{enumerate}

\numberwithin{equation}{section}
\newtheorem{theorem}{Theorem}[section]
\newtheorem{prop}[theorem]{Proposition}
\newtheorem{thm}[theorem]{Theorem}
\newtheorem{corollary}[theorem]{Corollary}
\newtheorem{lemma}[theorem]{Lemma}
\newtheorem{remark}[theorem]{Remark}

\newtheorem{defi}[theorem]{Definition}
\newtheorem{example}[theorem]{Example}
\newtheorem{assumption}[theorem]{Assumption}

\makeatletter
\newcommand{\myitem}[1]{%
	\item[(#1)]\protected@edef\@currentlabel{#1}%
}
\makeatother

\def\eatspace#1{#1}
\def\step#1#2{\par\kern1pt\hangindent#2em\hangafter=1\noindent\rlap{\small#1}\kern#2em\relax\eatspace}

\let\set\mathbb
\def\<#1>{\langle#1\rangle}

\def\val{\operatorname{val}}

\def\diag{\operatorname{diag}}

\def\a{\alpha}

\begin{document}
\title{On the Existence of Telescopers for P-recursive Sequences
	\thanks{L.\ Du was supported by the Austrian FWF grant P31571-N32. 
	}}

\author{  \bigskip
	Lixin Du\\ \normalsize
	Institute for Algebra, Johannes Kepler University, Linz, A4040, Austria\\
	{\sf \normalsize lixindumath@gmail.com}\\ \\ 
}
\date{}
\maketitle
\vspace{-3em}
\begin{abstract}
	We extend the criterion on the existence of telescopers for hypergeometric terms to the case of P-recursive sequences. This criterion is based on the concept of integral bases and the generalized Abramov-Petkov\v sek reduction for P-recursive sequences.

\end{abstract}

\section{Introduction}
Given a sequence $f(t,x)$ in two variables, the task of creative telescoping is to find a nonzero recurrence operator~$T$, in $t$ only, such that 
\[T\cdot f(t,x) = g(t, x+1) - g(x),\]
where $g$ belongs to the same class of sequences as $f$. Such an operator $T$ is called a {\em telescoper} for $f$ and $g$ is called a {\em certificate} for $T$. In 1990, Zeilberger developed an algorithm to construct a telescoper for hypergeometric terms~\cite{Zeilberger1991,PWZbook1996}. Zeilberger's algorithm is a useful tool for proving combinatorial identities, which has been developed for various class of functions~\cite{Zeilberger1990,chyzak00,koutschan10b}. In the differential case, telescopers always exists for D-finite functions~\cite{Zeilberger1990}. In the shift case, this is not true. 

At the beginning, Wilf and Zeilberger~\cite{Wilf1992} proved that telescopers exist for proper hypergeometric terms. After ten years, Abramov and Le~\cite{AbramovLe2002} gave a necessary and sufficient condition on the existence of telescopers for rational functions using the additive decomposition~\cite{Abramov1995,Pirastu1995b}. More precisely, any rational function $f(t,x)$ can be decomposed as \[g(t,x+1)- g(t,x) + h(t,x),\] where $g$ and $h$ are rational functions and the denominator of $h$ has the minimal degree with respect to~$x$. The~existence criterion says that $f$ has a telescoper if and only if the denominator of $h$ is integer-linear. This criterion was soon extended to the hypergeometric case by Abramov~\cite{Abramov03} using Abramov-Petkov\v sek reduction~\cite{AbramovPetkovsek01,AbramovPetkovsek01b}.

For a more general class of functions, Zeilberger~\cite{Zeilberger1990} proved that telescopers always exist for holonomic functions. Chen, Kauers and Koutschan proposed the notion of proper $\partial$-finite functions, which gives a sufficient condition on the existence of telescopers for $\partial$-finite functions. For P-recursive sequences, there are several additive decomposition algorithms~\cite{vanderHoeven18b,BrochetBruno23,chen23b}, which have been used to develop reduction-based telescoping algorithms. They can compute a telescoper if it exists. If a telescoper does not exist, the algorithms will run into an infinite loop to search a telescoper. The termination of these algorithms is equivalent to the existence problem of telescopers, which is not completely solved although some sufficient conditions are given in~\cite{Zeilberger1990,chen14,vanderHoeven18b}. 

In this paper, we shall modify the notion of {\em proper} P-recursive sequences in~\cite{chen14}, and give a complete answer to the existence problem of telescopers for P-recursive sequences,  see Section~\ref{sec:proper}. The main result is as follows: a P-recursive sequence has a telescoper if and only if it can be decomposed into a summable part and a proper part (Theorem~\ref{Thm: main}). Similar to the hypergeometric case~\cite{Abramov03,Huang16}, our existence criterion is based on the generalized Abramov-Petkov\v sek reduction~\cite{chen23b} via integral bases. The proof of the criterion is given in Sections \ref{sec: proper implies existence} and~\ref{sec: existence implies proper}. For a proper P-recursive sequence, we give an upper bound on the order of minmial telescopers (Corollary~\ref{Cor: tele bound}).

\section{Proper P-recursive sequences}\label{sec:proper}
Let $C$ be a field of characteristic zero and $\bar C$ be the algebraic closure of $C$. Let $C(t,x)$ be the field of rational functions in $t$ and $x$ over $C$. Let $\sigma_t,\sigma_x\colon C(t,x)\to C(t,x)$ be field automorphisms such that $\sigma_t(f(t,x))=f(t+1,x)$ and $\sigma_x(f(t,x)) = f(t,x+1)$ for all $f\in C(t, x)$. Let $C(t,x)[S_t, S_x]$ be an Ore algebra~\cite{ChyzakSalvy98}, where $S_t$ and $S_x$ are shift operators satisfying the multiplication rules $S_tS_x = S_xS_t$, $S_tf=\sigma_t(f)S_t$ and $S_xf = \sigma_x(f)S_x$ for all $f\in C(t,x)$. The difference operators $S_t-1$ and $S_x-1$ are denoted by $\Delta_t$ and $\Delta_x$, respectively. 

Let $J$ be a left ideal of $C(t, x)[S_t, S_x]$, such that $A=C(t, x)[S_t, S_x]/J$ is a $C(t,x)$-vector space of finite dimension $r$. When there is no ambiguity, an equivalence class $f+J$ in $A$ is also denoted by~$f$. An element $f\in A$ is called {\em summable} in $A$ (with respect to $x$) if there exists $g\in A$ such that $f=\Delta_x(g)$. If $f\in A$ is not summable, one can ask to find a nonzero operator $T \in C(t)[S_t]$ (free of $x$) such that $Tf$ is summable in $A$. Such an operator $T$ if it exists is called a \emph{telescoper} for $f$ of type~$(S_t; S_x)$.

Let $W=(\omega_1,\ldots, \omega_r)$ be a basis of $A$ as a $C(t,x)$-vector space. Let $e, e_t\in C[t,x]$ and $M, M_t\in C[t,x]^{r\times r}$ be such that $S_xW = \frac{1}{e}MW$ and $S_tW = \frac{1}{e_t}M_tW$, where $S_xW=(S_x\omega_1,\ldots, S_x\omega_r)$ and $S_tW=(S_t\omega_1,\ldots, S_t\omega_r)$. In this paper, we make the following assumption.
\begin{assumption}\label{Assump}
	There exists a basis $W$ of $A$ such that the matrices $\frac{1}{e}M$ and $\frac{1}{e_t}M_t$ are invertible over $C(t,x)$.
\end{assumption}
 Let $\tilde W$ be another basis of $A$ and let $T\in C(t,x)^{r\times r}$ be an invertible matrix such that $\tilde W=T W$. Then $S_x\tilde W = \sigma_x(T)S_xW=\sigma_x(T)\frac{1}{e}MT^{-1}\tilde W$ and $S_t\tilde W=\sigma_t(T)S_tW=\sigma_t(T)\frac{1}{e_t}M_tT^{-1}\tilde W$, where the shift of $T$ is the matrix obtained by taking the shift of all its entries. So if we write $S_x\tilde W=\frac{1}{\tilde e}\tilde M\tilde W$ and $S_t \tilde W= \frac{1}{\tilde e_t}\tilde M_t\tilde W$, where $e,\tilde e\in C[t,x]$ and $M,\tilde M\in C[t,x]^{r\times r}$, then $\frac{1}{\tilde e}\tilde M$ and $\frac{1}{\tilde e_t}\tilde M_t$ are also invertible over $C(t,x)$. Therefore Assumption~\ref{Assump} can be checked by choosing an arbitrary basis of $A$.

 Under the above assumption that $\frac{1}{e}M$ is invertible, there exists a cyclic vector $\gamma$ with respect to $S_x$, see~\cite[Theorem B.2]{HendricksSinger99}. This means that $\{\gamma, S_x\gamma, \ldots, S_x^{r-1}\gamma\}$ is a basis of $A$ over $K(x)$ with $K=C(t)$. Then $\gamma$ is annihilated by $L$ and $S_t - u_t$ for some $L, u_t\in K(x)[S_x]$. Without loss of generality, we assume that $L=\ell_0 + \ell_1S_x+\cdots+\ell_rS_x^r\in K[x][S_x]$ with polynomial coefficients $\ell_i \in K[x]$ and $\ell_0\ell_r\neq 0$. Then $r$ is called the {\em order} of $L$. Every element $f$ in $A$ can be written as $P_f\gamma$ for some $P_f \in K(x)[S_x]$. From now on, we identify $A$ with $K(x)[S_x]/\<L>$ by the map sending $f$ to $P_f+\<L>$ if $f = P_f \gamma$ for some $P_f\in K(x)[S_x]$. It can be checked that this map is a $K(x)[S_x]$-module isomorphism.
 
An irreducible polynomial $p\in C[t,x]$ is called {\em integer-linear} over $C$ if there exist a univariate polynomial $h\in C[z]$ and two integers $m,n\in \set Z$ such that $p = h(mt+nx)$. A polynomial $p\in C[t,x]$ is called {\em integer-linear} over $C$ if all its irreducible factors are integer-linear. A rational function $f\in C(t,x)$ is called {\em integer-linear} over $C$ if its numerator and denominator are integer-linear. For a given polynomial $p$, there is an algorithm~\cite{AbramovLe2002} for deciding whether $p$ is integer-linear or not without computing its irreducible factorization.

\begin{defi}\label{Defi: proper}
	Let $W$ be a $K(x)$-vector space basis of $A$. Any $f\in A$ can be represented as
	\[f = \frac{aW}{u},\]
	where $a=(a_1,\ldots, a_r)\in C[t,x]^r$, $u\in C[t,x]$ and $\gcd(a_1,\ldots, a_r, u) = 1$. We say that $f$ is a {\em proper} P-recursive sequence (with respect to $x$) if there exists a suitable basis $W$ (with respect to $x$) at $\{\beta_1,\ldots, \beta_I\}$ for some $\beta_i\in \bar K$ such that the denominator $u$ is integer-linear. 
\end{defi}

We shall recall the definition~\cite{chen23b} of suitable bases in Section~\ref{sec: integral}. For a suitable basis $W$, if we write $S_xW=\frac{1}{e}MW$ with $e\in C[t,x]$ and $M\in C[t,x]^{r\times r}$, then $W$ has the feature that $e$ is shift-free with respect to $x$, i.e., $\gcd(e, \sigma_x^i(e)) = 1$ (as polynomials in $K[x]$) for all $i\in \set Z\setminus\{0\}$. Such a suitable basis exists and can be computed~\cite{chen23b}. We will prove in Corollary~\ref{Cor: indep of suitable} that the properness of P-recursive sequences is independent of the choice of suitable bases. Furthermore, we will show in Corollary~\ref{Cor: sym} that the properness with respect to $x$ is equivalent to the properness with respect to $t$. The classical definition of {\em proper} hypergeometric terms can be viewed as a special case of our definition, see Corollary~\ref{Cor: hypergeo}. 

The main results of this paper are as follows.
\begin{theorem}\label{Thm: main_shiftfree}
	Let $W$ be a suitable basis of $A$ at $\{\beta_1,\ldots, \beta_I\}$. Let $f = \frac{aW}{u}$, where $a=(a_1,\ldots, a_r)\in C[t,x]^r$, $u\in C[t,x]$ is shift-free with respect to $x$ and $\gcd(a_1,\ldots, a_r, u) = 1$. If $f$ has a telescoper of type $(S_t;S_x)$, then $u$ is integer-linear.
\end{theorem}

\begin{theorem}\label{Thm: main}
	A P-recursive sequence $f\in A$ has a telescoper of type $(S_t; S_x)$ if and only if 
	\[f = \Delta_x(g) +h\]
	where $g,h\in A$ and $h$ is a proper P-recursive sequence with respect to $x$.  
\end{theorem}

Let $W$ be a suitable basis of $A$ at $\{\beta_1,\ldots, \beta_I\}$. Applying the generalized Abramov-Petkov\v sek reduction for P-recursive sequences~\cite{chen23b} to $f$, see Section~\ref{sec: add}, one can decompose $f = \Delta_x(g)+h$, where $g,h\in A$ and the denominator $\tilde u$ of $h$ with respect to $W$ is shift-free with respect to $x$. Then $f$ has a telescoper of type $(S_t;S_x)$ if and only if $h$ has a telescoper of type $(S_t;S_x)$. By Theorems \ref{Thm: main_shiftfree} and~\ref{Thm: main}, $h$ has a telescoper of type $(S_t;S_x)$ if and only if $\tilde u$ is integer-linear. This gives an algorithm for deciding the existence of telescopers for a given P-recursive sequences. The proof of Theorems \ref{Thm: main_shiftfree} and~\ref{Thm: main} will be given later.

\begin{example}\label{Eg: proper1}
	We consider the existence of telescopers for rational functions. Let $F=\frac{x^2 +t}{(x-t)^2}$, which is annihilated by
	\[L = (x+1-t)^2(x^2+t)S_x - ((x+1)^2 + t)(x-t)^2\quad \text{and}\quad S_t -\frac{(x^2+t+1)(x-t)^2}{(x^2+t)(x-t-1)^2}.\]
	A suitable basis of $A = K(x)[S_x]/\<L>$ with $K=\set C(t)$ at $\{\beta_1,\beta_2,\beta_3\}$ is $W = \{\frac{(x-t)^2}{x^2 + t}\}$, where $\beta_1= t$, $\beta_2$ and $\beta_3$ are distinct roots of $x^2 + t$. 
	\begin{enumerate}[{(i)}]
		\item The function $F$ corresponds to $f= 1\in A$. Since $f = \frac{x^2 + t}{(x-t)^2} W$ with denominator $u = (x-t)^2$ being integer-linear, $f$ is proper. By Theorem~\ref{Thm: main}, $f$ has a telescoper of type $(S_t;S_x)$. Indeed, $T = -(t^2 + t) S_t^2 + (2t^2 + 4 t) S_t -(t^2+3t+2)$ is a telescoper for $f$.
		\item The function $\frac{1}{x^2 + t} = \frac{(x-t)^2}{(x^2 + t)^2} F$ corresponds to $f =\frac{(x-t)^2}{(x^2+t)^2}\in A$. Since $f =\frac{1}{(x^2 + t)}W$ with denominator $u=x^2 +t$ being shift-free with respect to $x$ but not integer-linear, it follows from Theorem~\ref{Thm: main_shiftfree} that $f$ does not have any telescoper of type $(S_t;S_x)$. 
	\end{enumerate} 
\end{example}
\begin{example}\label{Eg: proper2}
	Let $F=\binom tx^3$ be the same as in~\cite[Example 6.6]{chen15a}. Then $F$ is annihilated by
	\[L = (1+x)^3S_x - (t-x)^3\quad \text{and}\quad S_t - \frac{(1+t)^3}{(1+t-x)^3}.\] 
	A suitable basis of $A = K(x)[S_x]/\<L>$ with $K=\set C(t)$ at $\{-1,t\}$ is $W = \{1\}$. 
	The hypergeometric term $F$ corresponds to $f =1\in A$. Since $f=W$ with denominator $u=1$ being integer-linear, $f$ is proper. By Theorem~\ref{Thm: main}, $f$ has a telescoper of type $(S_t;S_x)$. Indeed, $T=(t+2)^2S_t^2 - (7 t^2+ 21 t +16 )S_t -8 (t+1)^2$ is  a telescoper for $f$.
\end{example}
\begin{example}\label{Eg: proper3}
	Let $F=x+t^2+\frac{1}{x!}$ be the same as in~\cite[Example 26]{chen23b}. Then $F$ is annihilated by
	\[L =(x+2)(x^2+(t^2+1)x+1)S_x^2 - (x^3+(t^2+5)x^2+(3t^2+7)x+t^2+4)S_x +x^2+(t^2+3)x+t^2+3\]
	and
	\[S_t -\frac{(2t+1)(x+1)}{x^2+(t^2+1)x+1}S_x- \frac{x^2+(t^2+2)x-2t}{x^2+(t^2+1)x+1}.\]
	A suitable basis of $A=K(x)[S_x]/\<L>$ with $K=\set C(t)$ at $\{\beta_1,\beta_2,\beta_3\}$ is
	\[W=(\omega_1,\omega_2)=\left(1,\frac{1}{x^2 + (t^2+1)x+1}\left(S_x - \frac{x+t^2}{t^2-1}\right)\right),\]
	where $\beta_1=-2$, $\beta_2$ and $\beta_3$ are distinct roots of $x^2+(t^2+1)x+1$. 
	\begin{enumerate}[(i)]
		\item The sequence $\frac{1}{x+t}F$ corresponds to $f =\frac{1}{x+t}\in A$. Its representation in the basis is
		\begin{equation*}
			f=\frac{1}{x+t}\omega_1= \frac{1}{x+t}(1,0)W.
		\end{equation*}
		Since its denominator $u=x+t$ is integer-linear, f is proper. By Theorem~\ref{Thm: main}, $f$ has a telescoper of type $(S_t;S_x)$. Indeed, $T=(3t^2+3t+2)S_t^3+(3t^3+3t^2-4t-6)S_t^2 -(6t^3+15t^2+13t+2)S_t+3t^3+9t^2+8t$ is a telescoper for $f$.
		\item The sequence $\frac{1}{x^2+t}F$ corresponds to $f =\frac{1}{x^2+t}\in A$. Since $f=\frac{1}{x^2+t}(1,0)W$ with denominator $u=x^2+t$ being shift-free with respect to $x$ but not integer-linear, it follows from Theorem~\ref{Thm: main_shiftfree} that $f$ does not have any telescoper of type $(S_t;S_x)$.
	\end{enumerate}
	
\end{example}
\section{Integral bases and suitable bases}\label{sec: integral}
Let $K$ be a field of characteristic zero. Throughout the paper, the field $K$ is $C(t)$. For each $\alpha\in \bar K$, every nonzero rational function $u\in K(x)$ can be written as $u = (x-\alpha)^m p/q$ for some $m\in \set Z$, $p,q\in \bar K[x]$ with $\gcd(p,q)=1$ and $(x-\alpha)\nmid pq$. The {\em valuation} of $\nu_\alpha(u)$ of $u$ at $\alpha$ is defined as the integer $m$. The {\em valuation} $\nu_\infty(f)$ of a nonzero rational function $u= p/q$ with $p,q\in K[x]$ is defined as $\nu_\infty(u) = \deg_x(q)- \deg_x(p)$. For each $\alpha\in \bar K\cup\{\infty\}$, set $\nu_\alpha(0)=\infty$. The field $K(x)$ equipped with a valuation $\nu_\alpha$ is a {\em valued field}. The set $K(x)_\alpha :=\{u \in K(x)\mid \nu_\alpha(u)\geq 0\}$ forms a subring of $K(x)$.

Let $L=\ell_0 + \ell_1S_x +\cdots + \ell_rS_x^r \in K[x][S_x]$ with $\ell_i\in K[x]$ and $\ell_0\ell_r\neq 0$. Then $A = K(x)[S_x]/\<L>$ a $K(x)$-vector space of dimension $r$. A map $\val_\alpha
\colon A \to \set R\cup\{\infty\}$ is called a {\em value function} on $A$ at $\alpha$ if for all $f,g\in A$ and $u \in K(x)$, it satisfies three properties: (\romannumeral 1) $\val_\alpha(f)= \infty$ if and only if $f=0$; (\romannumeral 2)~$\val_\alpha(uf)=\nu_\alpha(u)+\val_\alpha(f)$; (\romannumeral 3) $\val_\alpha(f+g)\geq \min\{\val_\alpha(f), \val_\alpha(g)\}$. For each $\alpha\in \bar K\cup \{\infty\}$, such a value function exists. Let $\val_\alpha$ be the value function introduced in~\cite{chen20,chen23b}. An element $f\in A$ is called {\em (locally) integral} at $\alpha$ if $\val_\alpha(f)\geq 0$. For a subset $Z\subseteq \bar K$, an element $f\in A$ is called {\em(locally) integral} at $Z$ if $f$ is locally integral at all $\alpha\in Z$. 

The set of all elements $f\in A$ that are locally integral at some fixed $\alpha\in \bar K\cup\{\infty\}$ forms a $ K(x)_\alpha$-module. A basis of this module is called a {\em local integral basis} at $\alpha$ of $A$. For $Z\subseteq \bar K\cup \{\infty\}$, a basis of~$A$ is called a {\em local integral basis} at $Z$ if it is a {\em local integral basis} at all $\alpha\in Z$. Different from the D-finite case, there may not exist a basis that is a local integral basis of $A$ at all $\alpha \in \bar K$. If $Z$ is a finite subset of $\bar K\cup\{\infty\}$, a local integral basis of $A$ at $Z$ exists and can be computed~\cite{chen20,chen23b}.

A $K(x)$-vector space basis $\{\omega_1,\dots,\omega_r\}$ of $A$ is
called \emph{normal at infinity} if there exist integers $\tau_1,\dots,\tau_r\in \set Z$ such that $\{x^{\tau_1}\omega_1,\dots,x^{\tau_r}\omega_r\}$ is a local integral basis at infinity. Given a basis of $A$, Trager's algorithm~\cite{Trager84} for algebraic function fields can be literally adapted to compute a basis of $A$ that is normal at infinity and generates the same $K[x]$-module as the given basis.

As introduced in~\cite{chen23b}, a suitable basis of $A$ is a local integral basis at some finite subset $Z\subseteq \bar K$. The construction of $Z$ is as follows. Since $\ell_0\ell_r\neq 0$, let $\a_1+\set Z, \ldots,\a_I +\set Z$ be distinct sets such that for each $i$ with $1\leq i\leq I$, the product $\ell_0\ell_r$ has at least one root in $\a_i +\set Z$. For $\beta, \gamma\in \a+\set Z$ with $\a\in \bar K$, we say $\beta<\gamma$ if $\beta-\gamma<0$. For each orbit $\a_i + \set Z$, let $\a_{i,1}<\a_{i,2}<\ldots <\a_{i,J_i}$ be all roots of $\ell_0\ell_r$ in $\a_i+\set Z$. We choose an arbitrary $\beta_i\in \a_i+\set Z$ such that $\alpha_{i, J_i}\leq \beta_i$. A basis of $A$ is called {\em suitable} at $\{\beta_1,\ldots, \beta_I\}$ if it is a local integral basis of $A$ at $Z=\bigcup_{i=1}^I\{\gamma\in \a_i+\set Z \mid \a_{i,1}<\gamma\leq\beta_i\}$ and it generates the same $K(x)_\a$-module as $U=\{1,S_x,\ldots, S_x^{r-1}\}$ for all $\a\in \bar K \setminus Z$. For simplicity, a basis is called {\em suitable} if it is suitable at some $\{\beta_1,\ldots, \beta_I\}$. In order to obtain a suitable basis of $K(x)[S_x]/\<L>$ (instead of $\bar K(x)[S_x]/\<L>$), we assume that if $\alpha_{i,J_i}$ and $\alpha_{i,J_{i'}}$ are conjugate over $K$ for some $1\leq i\neq i'\leq I$, then $\beta_i-\alpha_{i,J_i} = \beta_{i'} - \alpha_{i,J_{i'}}$. In this case, for each $1\leq i\leq I$, all conjugates of $\beta_i$ over $K$ belong to $\{\beta_1,\ldots, \beta_I\}$.

\begin{theorem}\label{Thm: suitable}
	Let $W$ be a suitable basis of $K(x)[S_x]/\<L>$ at $\{\beta_1,\ldots,\beta_I\}$. Let $e\in K[x]$ and $M\in K[x]^{r\times r}$ be such that $S_xW = \frac{1}{e}MW$. Then 
	\begin{enumerate}[(i)]
		\item\label{it:suitable1} $e$ is shift-free with respect to $x$ and $\beta_1, \ldots, \beta_I$ are all possible roots of $e$ in $\bar K$;
		\item\label{it:suitable2} 
		$\beta_1, \ldots, \beta_I$ are all possible roots and poles of $\det(\frac{1}{e}M)$ in $\bar K$;
		\item\label{it:suitable3} $\det(M)$ is shift-free with respect to $x$ and $\beta_1,\ldots, \beta_I$ are all possible roots of $\det(M)$ in $\bar K$;
		\item\label{it:suitable4} $\gcd(\det(M), \sigma_x^i(e))=1$ for all $i\in \set Z\setminus\{0\}$.
	\end{enumerate}
\end{theorem}
\begin{proof}
	\eqref{it:suitable1} Done in~\cite[Theorem 12.(\romannumeral 1)]{chen23b}.
        
   \eqref{it:suitable3} Since $\det(M) = e^r\det(\frac{1}{e}M)$, it follows from 	\eqref{it:suitable1} and \eqref{it:suitable2} that $\beta_1,\ldots, \beta_I$ are all possible roots of $\det(M)$. Since $\beta_1,\ldots,\beta_I$ belong to distinct cosets of $\set Z$ in $\bar K$, the difference of any two roots of $\det(M)$ in $\bar K$ can not be an integer. So $\det(M)$ is shift-free with respect to $x$.
   
   \eqref{it:suitable4} follows from 	\eqref{it:suitable1} and \eqref{it:suitable3} because $\beta_1,\ldots,\beta_I$ belong to distinct $\set Z$-orbits. 

	
	\eqref{it:suitable2} Let $U=\{1,S_x,\ldots, S_x^{r-1}\}$. A direct calculation yields that
	\[S_xU=\begin{pmatrix}
		0      & 1 &    &  \\
		\vdots &   & \ddots &  \\
		0      &   &        &1\\
		-\frac{\ell_0}{\ell_r} & -\frac{\ell_1}{\ell_r}&\cdots & -\frac{\ell_{r-1}}{\ell_r}
	\end{pmatrix} U =: \frac{1}{e_u}M_uU,
	\]
	where $e_u \in K[x]$ and $M_u\in K[x]^{r\times r}$. So $\det(\frac{1}{e_u}M_u) = (-1)^{r}(\ell_0/\ell_r)$. Let $T\in K(x)^{r\times r}$ be such that $W=TU$. Taking the shift operation of both sides, we get
	\begin{equation*}
		S_xW= \sigma_x(T)S_xU=\sigma_x(T)\frac{1}{e_u}M_uU=\sigma_x(T) \frac{1}{e_u}M_uT^{-1}W=\frac{1}{e}MW.
	\end{equation*}
	Since $W$ is a basis, we have $\frac{1}{e}M=\sigma_x(T)\frac{1}{e_u}M_uT^{-1}$. Note that $\sigma_x$ is a ring homomorphism. It follows that
	\begin{equation}\label{EQ: suitable det}
		\det\left(\frac{1}{e}M\right) = \sigma_x(\det(T)) \det\left(\frac{1}{e_u}M_u\right) \det(T)^{-1}
	\end{equation}
	Let $Z=\bigcup_{i=1}^I\{\gamma\in\alpha_i+\set Z\mid \alpha_{i,1}<\gamma\leq \beta_i\}$. Since $W$ and $U$ generate the same $K(x)_\alpha$-module for all $\a \in \bar K\setminus Z$, both $\det(T)$ and $\det(T)^{-1}$ are invertible elements in $K(x)_\a^{r\times r}$. 
	So $\sigma_x(\det(T))$ is invertible in $K(x)_\alpha$ for all $\alpha \in  \bar K\setminus Z'$, where $Z'=\{\gamma -1 \mid \gamma \in Z\}=\bigcup_{i=1}^I\{\gamma\in \a_i+\set Z \mid \a_{i,1}\leq \gamma \leq \beta_i-1\}$. By~\eqref{EQ: suitable det} we have
	\begin{align}\label{EQ:roots of e_w}
		&\{\text{roots and poles of $\det(\tfrac{1}{e}M)$} \}\nonumber\\
		\subseteq &\{\text{roots and poles of $\sigma_x(\det(T)),\det(\tfrac{1}{e_u}M_u), \det(T)^{-1}$}\}\nonumber\\
		\subseteq& Z' \,\cup \,\{\text{roots of $\ell_0\ell_r$}\}\, \cup\, Z\nonumber\\
		\subseteq &\bigcup_{i=1}^I\left( \left\{\a_{i,1},\ldots, \beta_i-1\right\}\,\cup\,\{\a_{i,1},\a_{i,2},\ldots, \a_{i,J_i}\}\,\cup \,\{\a_{i,1}+1,\ldots, \beta_i\}\right)\nonumber\\
		\subseteq&\bigcup_{i=1}^I\,\{\a_{i,1},\a_{i,1}+1,\ldots, \beta_i\};
	\end{align}
	here we use the assumption that $\a_{i,J_i}\leq\beta_i$.
	
	By the same argument as in~\cite[Theorem 12.(\romannumeral 1)]{chen23b}, both $S_xW$ and $W$ are local integral bases at $Z'$. So~$\frac{1}{e}M$ is an invertible matrix over $K(x)_\alpha$ for all $\alpha\in Z'$ and hence $\det(\frac{1}{e}M)$ has neither a root nor a pole in $Z'$. Combining this fact and the relation in~\eqref{EQ:roots of e_w}, we get
	\[\{\text{roots and poles of $\det(\tfrac{1}{e}M)$} \}\subseteq\{\beta_i \mid 1\leq i\leq I\}.\] 
\end{proof}
The following corollary gives the same result as in~\cite[Theorem 12.(\romannumeral 2)]{chen23b} with a different proof.
\begin{corollary}\label{Cor: tilde e integer-linear}
	Let $W$ be a suitable basis of $K(x)[S_x]/\<L>$ at $\{\beta_1,\ldots,\beta_I\}$. Let $\tilde e\in K[x]$ and $\tilde M\in K[x]^{r\times r}$ be such that $W = \frac{1}{\tilde e}\tilde MS_xW$. Then $\tilde e$ divides $\det(M)$ and $\tilde e$ is shift-free with respect to $x$. Moreover $\beta_1, \ldots, \beta_I$ are all possible roots of $\tilde e$ in $\bar K$.
\end{corollary}
\begin{proof}
	Let $e\in K[x]$ and $M\in K[x]^{r\times r}$ be such that $S_xW= \frac{1}{e}M W$. Then $\frac{1}{\tilde e}\tilde M  = (\frac{1}{e}M)^{-1} = eM^{-1} = \frac{e}{\det(M)}M^*$, where $M^*$ is the adjoint matrix of $M$. So $\tilde e\mid \det(M)$ and the conclusion follows from Theorem~\ref{Thm: suitable}.\eqref{it:suitable3}.
\end{proof}
\section{Additive decomposition for univariate P-recursive sequences}\label{sec: add}
The Abramov-Petkov\v sek reduction~\cite{AbramovPetkovsek01,AbramovePetkovsek02} for hypergeometric terms was extended to the case of P-recursive sequences in two ways: one version relies on Lagrange's identity~\cite{vanderHoeven18b,BrochetBruno23}, another on integral bases~\cite{chen23b}. In this section, we recall the additive decomposition~\cite{chen23b} for univariate P-recursive sequences via integral bases. It decomposes a P-recursive sequence $f$ in $A= K(x)[S_x]/\<L>$ into the form $f = \Delta_x(g) + h$ such that 
$f$ is summable in $A$ if and only if $h = 0$.

The generalized Abramov-Petkov\v sek reduction formula~\cite[Lemma 14]{chen23b} is as follows.
\begin{lemma}\label{Lem: AP step}
	Let $W$ be a suitable basis of $A$ at $\{\beta_1,\ldots,\beta_I\}$. Let $e,\tilde e\in K[x]$ and $M,\tilde M\in K[x]^{r\times r}$ be such that $S_xW=\frac{1}{e}MW$ and $\frac{1}{\tilde e}\tilde M=(\frac{1}{e}M)^{-1}$. Let $p_i\in K[x]$ be the minimal polynomial of $\beta_i$ over $K$. Let $q\in K[x]$, $a\in K[x]^r$ and $\ell\in\set Z$.
	\begin{enumerate}[(i)]
		\item\label{it:AP-red1} If $\gcd(q, \sigma_x^j(p_i))=1$ for all $i\in\{1,\ldots, I\}$ and $j\in \set Z$, there exist $g\in A$ and $c\in K[x]^r$ such that 
		\begin{equation}\label{EQ: AP-step}
			\frac{aW}{\sigma_x^\ell(q)}=\Delta_x(g) +  \frac{cW}{qe\tilde e}.
		\end{equation}
		\item\label{it:AP-red2} If $q=p_i^m$ with $i\in\{1,\ldots, I\}$ and $m>0$, there exist $g\in A$ and $c\in K[x]^r$ such that~\eqref{EQ: AP-step} also holds.
	\end{enumerate}
\end{lemma}

Let $W$ be a suitable basis of $A$ at $\{\beta_1,\ldots, \beta_I\}$. Using the above lemma, we can decompose any element $f\in A$ into the form
\begin{equation}\label{EQ: AP-reduction}
	f = \Delta_x(g)+ h\quad\text{and}\quad h=\frac{cW}{\tilde u},
\end{equation}
where $g\in A$, $c\in K[x]^r$, $\tilde u\in K[x]$ and the product $\tilde u\prod_{i=1}^I(x-\beta_i)$ is shift-free with respect to $x$.

\begin{theorem}\label{Thm: AP-remainder}
	Let $h\in A$ be in~\eqref{EQ: AP-reduction} with $c=(c_1,\ldots,c_r)$ and $\gcd(\tilde u, c_1,\ldots,c_r)=1$. Let $e\in K[x]$ and $M\in K[x]^{r\times r}$ be such that $S_xW = \frac{1}{e}MW$. If $h$ is summable in $A$, then  $\tilde u$ divides $e$ and $h=\Delta_x(bW)$ with $b\in K[x]^r$.
\end{theorem}

\begin{corollary}\label{Cor: AP-remainder}
	 The decomposition~\eqref{EQ: AP-reduction} can be rewritten as
	 \begin{equation}\label{EQ: AP-remainder}
	 	f = \Delta_x(g)+h\quad\text{and} \quad h= \frac{1}{d}PW+\frac{1}{e}RW,
	 \end{equation}
	 where $d\in K[x]$, $P,R\in K[x]^r$ with $\deg_x(P)<\deg_x(d)$ and $d\prod_{i=1}^I(x-\beta_i)$ being shift-free with respect to~$x$. Moreover, if $f$ is summable in $A$, then $d\in K$ and $P=0$.
\end{corollary}

Let $W$ be a suitable basis of~$A$ at $\{\beta_1,\ldots, \beta_I\}$ that is normal at infinity. Let $T = \diag\bigl(x^{\tau_1}, \ldots, x^{\tau_r}\bigr) \in K(x)^{r\times r}$ with $\tau_i\in \set Z$ be such that $V := TW$ is a local integral basis at infinity. Let $e, a\in K[x]$ and $M, B\in K[x]^{r \times r} $ be such that $S_xW = \frac{1}{e}MW$ and $\Delta_x V=\frac{1}{a}BV$. Multiplying both $a$ and $B$ by some polynomials, we may assume that $e$ divides $a$. For $\mu, \delta\in \set Z$ with $\mu \leq \delta$, the subspace $\{\sum_{i=\mu}^{\delta} a_i x^i\,|\,a_i\in K\}$ of Laurent polynomials in $K[x,x^{-1}]$ is denoted by $K[x]_{\mu,\delta}$.

\begin{theorem}\label{Thm:add decomp}
	Let $W, V\in A^r$ be as described above. Then $\frac{1}{e}RW$ in~\eqref{EQ: AP-remainder} can be further reduced and any element $f\in A$ can be decomposed into
	\begin{equation}\label{EQ:add}
		f = \Delta_x(\tilde g) + \frac{1}{d} PW + \frac{1}{a} QV,
	\end{equation}
	where $\tilde g\in A$, $d\in K[x]$, $P\in K[x]^r$ with $\deg_x(P) < \deg_x(d)$ and $Q\in N_V$ such that 
	\begin{enumerate}[(i)]
		\item the product $d\prod_{i=1}^I(x-\beta_i)$ is shift-free with respect to $x$;
		\item $N_V$ is a subspace of $K[x]_{\mu,\delta}^r$ with  $\mu = \min\{-\tau_1, \ldots, -\tau_r, 0\}$ and $\delta = \max\{\deg_x(a),\deg_x(B)\}-1$;
		\item $f$ is summable in $A$ if and only if $P=Q=0$.
	\end{enumerate}
\end{theorem}
\section{Integer-linear polynomials}\label{sec:integer-linear}
Let $G=\<\sigma_t,\sigma_x>$ be the free abelian group generated by $\sigma_t$ and $\sigma_x$ with the multiplication being the composition of shift operators.
Let $p\in C[t,x]$ and $H$ be a subgroup of $G$. The set
\[[p]_H:=\{\sigma(p)\mid \sigma\in H\}\]
is called the {\em$H$-orbit} at $p$. Two irreducible polynomials $p,q\in C[t,x]$ are said to be {\em $H$-shift equivalent}, denoted by $p\sim_H q$ if $[p]_H=[q]_H$, i.e., there exists $\sigma\in H$ such that $q=\sigma(p)$. The relation $\sim_H$ is an equivalence relation. A polynomial $p\in C[t,x]$ is said to be {\em monic} if the leading coefficient of $p$ is $1$ with respect to the pure lexicographic order $t\prec x$. 
\begin{lemma}\label{Lem: integer-linear lem3}
	Let $p\in C[t,x]$ and $m,n\in \set Z$ with $n\neq0$. Then  $\sigma_t^n\sigma_x^m(p) = p$ if and only if $p=h(mt-nx)$ for some univariate polynomial $h\in C[z]$.
\end{lemma}
\begin{proof}
	If $p=h(mt-nx)$ with $h\in C[z]$, then $\sigma_t^n\sigma_x^m(p) = h(mt+mn-nx-nm)=h(mt-nx)=p$. For the proof of the other direction, see~\cite[Lemma 3]{AbramovPetkovsek01b}.
\end{proof}
\begin{lemma}(See~\cite[Theorem 6]{AbramovPetkovsek01b})\label{Lem: integer-linear thm6}
	Let $p\in C[t,x]$. If for every irreducible factor $q$ of $p$, there are $m, n\in \set Z$, $n>0$ such that $\sigma_t^n\sigma_x^m(q)$ divides $p$, then $p$ is integer-linear.
\end{lemma}
\begin{proof}
	By a repeated use of the assumption, for every irreducible factor $q$ of~$p$, there exist $m_i, n_i\in \set Z$, $0<n_1<n_2<\cdots$ such that $q$, $\sigma_t^{n_1}\sigma_x^{m_1}(q)$, $\sigma_t^{n_2}\sigma_x^{m_2}(q),\cdots$ are irreducible factors of $p$. Since $p$ has only finitely many irreducible factors, it follows that $\sigma_t^n\sigma_x^m(q) = q$ for some $m,n\in \set Z$, $n>0$. By Lemma~\ref{Lem: integer-linear lem3}, $q$ is integer-linear and so is $p$.
\end{proof}
 
Let $p\in C[t,x]$ be an irreducible integer-linear polynomials. Then $p=h(mt + nx)$ for some integers $m,n\in \set Z$ and $h\in C[z]$. Without loss of generality, we may assume that $\gcd(m,n)=1$ and $n\geq 0$. By B{\' e}zout's relation, there exist unique integers $a,b\in \set Z$ with $0\leq b < n$ such that $a m + bn = 1$. The shift operator $\sigma_t^a \sigma_x^b$ is denoted by $\tau^{(m, n)}$, or just $\tau$ for brevity if $(m,n)$ is clear. Note that $\tau(h(mt+nx))=h(mt+nx+1)$. So for any two integers $i,j\in \set Z$, we have $\sigma_t^i\sigma_x^j(p) = \tau^{mi+nj}(p)$. 
For any irreducible polynomial $q\in C[t,x]$, if $p$ and $q$ are $\<\sigma_t,\sigma_x>$-shift equivalent, then $q= \tau^{\ell}(p)$ for some $\ell \in \set Z$. The set $\{\tau^\ell(p) \mid \ell \in \set Z\}$ consists of all irreducible polynomials that are $\<\sigma_t,\sigma_x>$-shift equivalent to $p$. For any integer $\ell\in\set Z$, there exist unique integers $i,j\in\set Z$ such that 
$a\ell = i + nj$ and $0\leq i \leq n-1$. Note that $\sigma_t^{n}(p) = \sigma_x^{-m}(p)$. Then 
\begin{equation}\label{EQ:tau p}
	\tau^\ell(p) = \sigma_t^{a\ell}\sigma_x^{b\ell}(p) = \sigma_t^{i+nj}\sigma_x^{b\ell}(p)=\sigma_t^i\sigma_x^{b\ell-mj}(p).
\end{equation}
So for any $\ell \in \set Z$, $\tau^\ell(p)$ is $\<\sigma_x>$-shift equivalent to $\sigma_t^{i}(p)$ for some $i$ with $0\leq i\leq n-1$.
For a Laurent polynomial $\xi=\sum_{i=\ell}^\rho k_i\tau^i\in \set N [\tau,\tau^{-1}]$ with $\ell,\rho \in \set Z$, $k_i\in \set N$ and $\ell \leq \rho$, define
\[p^{\xi} := \tau^\ell(p^{k_\ell})\tau^{\ell+1}(p^{k_{\ell+1}})\cdots \tau^\rho(p^{k_\rho}).\]

For any integer-linear polynomial $q\in C[t,x]$, by grouping together $\<\sigma_t,\sigma_x>$-shift-equivalent factors, it has a factorization of the form
\begin{equation}\label{EQ: shiftless decomp}
	q = c\ q_1^{\xi_1}\cdots q_s^{\xi_s},
\end{equation}
where $c\in C$, $s\in \set N$, $q_i\in C[t,x]$ is a monic irreducible integer-linear polynomial of the form $q_i=h_i(m_it+n_i x)$ for $h_i\in C[z]$, $m_i,n_i\in \set Z$ with $n_i\geq 0$, $\gcd(m_i,n_i)=1$, and $\xi_i\in\set N[\tau_i,\tau_i^{-1}]$ with $\tau_i = \tau^{(m_i,n_i)}$. Moreover,  $q_i$ and $q_{j}$ are not $\<\sigma_t,\sigma_x>$-shift equivalent whenever $i\neq j$. More details about such a factorization can be found in~\cite[Section 5.1]{Huang16}.

\begin{defi}\label{Defi: norm}
	Let $\tau\in \<\sigma_t,\sigma_x>$ and $\xi = \sum_{i=\ell}^{\rho}k_i\tau^i\in\set N[\tau,\tau^{-1}]$ with $\ell,\rho\in \set Z$, $k_i\in\set N$ and $\ell\leq \rho$. The maximal coefficient of $\xi$ is denoted by
	\[||\xi||_{\tau} := \max\{k_i \mid i = \ell, \ldots, \rho\}.\]
	The sum of all coefficients of $\xi$ is denoted by
	\[||\xi||_{\tau}^* := k_\ell + k_{\ell + 1} + \cdots + k_{\rho}.\]
\end{defi}
\begin{lemma}\label{Lem: norm}
	Let $\tau \in \<\sigma_t,\sigma_x>$ and $\xi, \xi_1,\xi_2\in\set N[\tau,\tau^{-1}]$. Let $\sigma = \tau^n$ with $n>0$. Then
	\begin{enumerate}[(i)]
		\item\label{it:norm ineq}
		$||\xi_1+\xi_2||_\tau \leq ||\xi_1||_\tau + ||\xi_2||_\tau$;
		\item\label{it:norm} $||\sigma^j \xi ||_\tau = || \xi ||_\tau$ for all $j\in \set Z$;
		\item\label{it:norm*} $||\sum_{j=0}^{s}\sigma^j\xi ||_\tau \leq || \xi ||_\tau^*$ for all $s\in \set Z$, $s\geq 0$.
	\end{enumerate}
\end{lemma}
\begin{proof}
	\eqref{it:norm ineq} We write $\xi_1 = \sum_{i\in \set Z} k_i^{(1)}\tau^i$ and $\xi_2 = \sum_{i\in \set Z} k_i^{(2)}\tau^i$, where $k_i^{(1)}, k_i^{(2)}\in \set N$ and only finitely many $k_i^{(1)}$ and $k_i^{(2)}$ are nonzero. Then $\xi_1+\xi_2 = \sum_{i\in\set Z} (k_i^{(1)}+k_i^{2})\tau^i$ and hence \[||\xi_1+\xi_2||_\tau = \max_{i\in\set Z}\{k_i^{(1)}+k_i^{(2)}\} \leq \max_{i\in\set Z} \{k_i^{(1)}\} +\max_{i\in \set Z} \{k_i^{(2)}\} = ||\xi_1||_\tau + ||\xi_2||_\tau.\]
	
	\eqref{it:norm} We write $\xi = \sum_{i\in \set Z} k_i\tau^i$, where $k_i\in \set N$ and only finitely many $k_i$ are nonzero. Since $\sigma = \tau^n$, we have $\sigma^j \xi = \sum_{i\in\set Z} k_i \tau^{nj+i }$. So $||\sigma^j \xi||_\tau = \max_{i\in \set Z} \{k_i \}= ||\xi ||_\tau$.
	
	\eqref{it:norm*} We write $\xi = \tau^\ell \sum_{i= 0}^\infty k_i\tau^i$, where $\ell\in \set Z$, $k_i\in\set N$ and only finitely many $k_i$ are nonzero. Then
	\[\sum_{j=0}^s\sigma^j\xi = \sum_{j=0}^s \tau^{nj}\tau^\ell\sum_{i= 0}^\infty k_i\tau^i = \tau^{\ell} \sum_{j=0}^s\sum_{i= 0}^\infty k_i \tau^{nj+i}=\tau^\ell \sum_{i= 0}^\infty\left(\sum_{j=0}^{\min\{\lfloor i/n\rfloor,s \} } k_{i-nj} \right)\tau^i.\]
	For any $i\geq0$, 
	\[\sum_{j=0}^{\min\{\lfloor i/n\rfloor,s \} } k_{i-nj} \leq \sum_{j=0}^{\infty} k_j = ||\xi||_\tau^*.\]
	So $|| \sum_{j=0}^s\sigma^j\xi||_\tau \leq  ||\xi||_\tau^*$.
\end{proof}
\section{Compatible matrices}\label{sec: compatibility}
Let $A=C(t,x)[S_t,S_x]/J$ be as described in Section~\ref{sec:proper}. Let $W$ be a suitable basis of $A\cong K(x)[S_x]/\<L>$ with respect to $x$, where $K=C(t)$. We write
\begin{equation}\label{EQ:shift matrix rep}
	S_xW=\frac{1}{e}MW\quad\text{and}\quad S_tW=\frac{1}{e_t}M_tW,
\end{equation}
where $e,e_t\in C[t,x]$, $M, M_t\in C[t,x]^{r\times r}$ and the greatest common divisor of all entries of $M$ (resp. $M_t$) and $e$ (resp. $e_t$) is $1$. By Assumption~\ref{Assump}, $\det(M)$ and $\det(M_t)$ are nonzero.
\begin{defi}
	Two matrices $F,G\in C(t,x)^{r\times r}$ are {\em compatible} if they satisfy
	$\sigma_t(F)G=\sigma_x(G)F$.
\end{defi}


\begin{theorem}\label{Thm: compatibility}
	Let $W$ be a suitable basis of $A$ with respect to $x$, which satisfies~\eqref{EQ:shift matrix rep}. Then
	\begin{enumerate}[(i)]
		\item\label{it:compat1} the matrices  $\frac{1}{e}M$ and $\frac{1}{e_t}M_t$ are compatible;
		\item\label{it:compat2} the polynomials $e,\det(M), e_t, \det(M_t)$ are integer-linear.
	\end{enumerate} 
\end{theorem}
\begin{proof}
	\eqref{it:compat1} 	By a direct calculation, we have 
	\[S_tS_x W = S_t\frac{1}{e}MW = \sigma_t\left(\frac{1}{e}M\right)\frac{1}{e_t}M_tW\quad\text{and}\quad S_xS_tW = S_x\frac{1}{e_t}M_tW=\sigma_x\left(\frac{1}{e_t}\right)\frac{1}{e}MW.\]
	Since $S_tS_xW=S_xS_tW$ and $W$ is a basis of $A$, it follows that 
	\begin{equation}\label{EQ: compatible}
		\sigma_t\left(\frac{1}{e}M\right)\frac{1}{e_t} M_t= \sigma_x\left(\frac{1}{e_t}M_t\right)\frac{1}{e}M.
	\end{equation}
	
	\eqref{it:compat2} Since $W$ is a suitable basis, it follows from Theorem~\ref{Thm: suitable}.\eqref{it:suitable4} that $\gcd(\det(M), \sigma_x^i(e)) = 1$ (w.r.t.~$x$) for all $i\in \set Z\setminus \{0\}$. Using the compatibility condition~\eqref{EQ: compatible}, we shall prove that $e,\det(M),e_t,\det(M_t)$ are integer-linear. To do this, we need several lemmas.
\end{proof}

\begin{example}\label{Eg: compat}
	\begin{enumerate}[(i)]
		\item Let $W$ be the same as in Example~\ref{Eg: proper1}. Then $S_xW=W$ and $S_tW=W$. So $e=\det(M)=e_t=\det(M_t)=1$. They are integer-linear.
		\item Let $W$ be the same as in~Example~\ref{Eg: proper2}. Then $S_x W = \frac{(1+x)^3}{(t-x)^3}W$ and $S_tW= \frac{(1+t)^3}{(t-x-1)^3}W$. So $e=(t-x)^3$, $\det(M)=(1+x)^3$, $e_t=(1+t-x)^3$ and $\det(M_t)=(1+t)^3$ are integer-linear.
		\item\label{it:compat_eg3} Let $W=(\omega_1,\omega_2)$ be the same as in~Example~\ref{Eg: proper3}. Then
		\begin{equation*}
			\begin{pmatrix}
				S_x\omega_1\\
				S_x\omega_2
			\end{pmatrix} = \frac{1}{(x+2)(t^2-1)^2}\begin{pmatrix}
				(x+t^2)(x+2)(t^2-1)&(x^2+(t^2+1)x+1)(x+2)(t^2-1)^2\\
				-x-1&-(x^2+2x-t^2+2)(t^2-1)
			\end{pmatrix}\begin{pmatrix}
				\omega_1\\
				\omega_2
			\end{pmatrix}		
		\end{equation*}
		and
		\[\begin{pmatrix}
			S_t\omega_1\\
			S_t\omega_2
		\end{pmatrix} =\frac{1}{t(t+2)(t^2-1)} \begin{pmatrix}
			t^2(t+2)^2&t(x+1)(2t+1)(t+2)(t^2-1)\\
			0&(t^2-1)^2
		\end{pmatrix}\begin{pmatrix}
			\omega_1\\
			\omega_2
		\end{pmatrix}.\]
		So $e=(x+2)(t^2-1)^2$, $\det(M)= (t^2-1)^4 (x+2)$, $e_t=t(t+2)(t^2-1)$ and $\det(M_t) =  t^2 (t+2)^2(t^2-1)^2$ are integer-linear.
	\end{enumerate}
\end{example}
\begin{remark}
	Let $F=\det(\frac{1}{e}M)$ and $G=\det(\frac{1}{e_t}M_t)$. By~\eqref{EQ: compatible}, we have $\sigma_t(F)G = \sigma_x(G)F$. So $F$ and $G$ are compatible. By the structure theorem of compatible rational functions~\cite[Theorem 8]{AbramovPetkovsek01b}, $F$ and $G$ are integer-linear (because $F$ is shift-reduced with respect to $x$ by Theorem~\ref{Thm: suitable}.\eqref{it:suitable2}). So if we assume that $\gcd(\det(M),e^r)=\gcd(\det(M_t),e_t^r)=1$ (this is true for $r=1$), then $e,\det(M), e_t,\det(M_t)$ are also integer-linear. However, this assumption may fail when $r>1$ and we can not use this short argument to prove Theorem~\ref{Thm: compatibility}.\eqref{it:compat2} directly. For instance, in Example~\ref{Eg: compat}.\eqref{it:compat_eg3}, neither $\gcd(\det(M),e^2)=(x+2)(t^2-1)^2$ nor $\gcd(\det(M_t),e_t)=t^2(t+2)^2(t^2-1)^2$ is $1$. Since $F=
	\frac{1}{x+2}$ and $G=1$, we can not read the complete information about factors of $e,\det(M), e_t, \det(M_t)$ from $F,G$. In order to prove Theorem~\ref{Thm: compatibility}.\eqref{it:compat2}, we shall generalize the proof of~\cite[Theorem 8]{AbramovPetkovsek01b} to the case of higher dimension.
\end{remark}


For an irreducible polynomial $p\in C[x]$, the set $\{f/g \mid f, g\in C[x], \, p \nmid g\}$ is a subring of $C(x)$, denoted by $C[x]_{(p)}$.
\begin{lemma}\label{Lem: gcd}
	Let $p\in C[x]$ be an irreducible polynomial over $C$. Let $F\in C[x]^{r\times r}$ be an invertible matrix over the ring $C[x]_{(p)}$. Let $M=(m_{i,j})$ and $N=(n_{i,j})$ be two matrices in $C[x]^{r\times r}$ such that $N=MF$ or $N=FM$. If~$\gcd(m_{1,1},m_{1,2},\ldots, m_{r,r,},p)=1$, then $\gcd(n_{1,1},n_{1,2},\ldots,n_{r,r},p)=1$.
\end{lemma}
\begin{proof}
	Suppose that $\gcd(m_{1,1},m_{1,2},\ldots, m_{r,r}, p) = 1$. Then there exist $a_{i,j}$ and $a_0$ in $C[x]$ such that
	 \begin{equation}\label{EQ: gcd=1}
		\sum_{i=1}^r\sum_{j=1}^ra_{i,j}m_{i,j}+a_0 p = 1.
	\end{equation}
	Let $J_M$ be the ideal of $C[x]_{(p)}$ generated by $m_{1,1}, m_{1,2},\ldots, m_{r,r}$ and $p$. The relation in~\eqref{EQ: gcd=1} implies that $J_M=(1)$. 
	
	Let $J_N$ be the ideal of $C[x]_{(p)}$ generated by $n_{1,1}, n_{1,2},\ldots, n_{r,r}$ and $p$. 	Since $F\in C[x]^{r\times r}$ is invertible over $C[x]_{(p)}$, its inverse $F^{-1}$ belongs to $C[x]_{(p)}^{r\times r}$. So we have $M=NF^{-1}$ or $M=F^{-1}N$. In both cases, each $m_{i,j}$ is a linear combination of the $n_{i,j}$'s with coefficients in $C[x]_{(p)}$. So $J_M \subseteq J_N$. Then we have $J_N=(1)$ because $J_M=(1)$. There are $b_{i,j}$ and $b_0$ in $C[x]_{(p)}$ such that $
	\sum_{i=1}^r\sum_{j=1}^r b_{i,j}n_{i,j} + b_0p= 1$. Since $p$ is irreducible, it follows that $\gcd(n_{1,1},n_{1,2},\ldots,n_{r,r},p)=1$ or $p$. However, the common denominator of the $b_{i,j}$'s and $b_0$ is not divisible by $p$. So we have $\gcd(n_{1,1},n_{1,2},\ldots,n_{r,r},p)=1$.
%
\end{proof}

\begin{prop}\label{Prop: division}
	Let $e$ be a nonzero polynomial in $C[x]$ and $M=(m_{i,j}),F,G\in C[x]^{r\times r}$ be three invertible matrices over $C(x)$ with $\gcd(m_{1,1},m_{1,2},\ldots,m_{r,r},e)=1$ such that
	\begin{equation}\label{EQ: weak compatibility}
		\sigma_x\left(\frac{1}{e}M\right)F = G \frac{1}{e}M \quad\text{or} \quad F\sigma_x\left(\frac{1}{e}M\right) = G\frac{1}{e}M.
	\end{equation}
	\begin{enumerate}[(i)]
		\item\label{it:division1} If $p$ is an irreducible factor of $e$, then there exist $m,n\in \set N$, $m\geq 1$, $n\geq 0$, such that $\sigma_x^m(p)$ divides $\det(F)$ and $\sigma_x^{-n}(p)$ divides $\det(G)$.
		\item\label{it:division2} If $p$ is an irreducible factor of $\det(M)$ and $p$ does not divide $e$, then there exist $m,n\in\set N$, $m\geq 1$, $n\geq 0$, such that $\sigma_x^m(p)$ divides $\det(G)$ and $\sigma_x^{-n}(p)$ divides $\det(F)$.
	\end{enumerate}
\end{prop}
\begin{proof}
	\eqref{it:division1} Suppose $p\mid e$. Let \[m = \max\{i\in \set N\mid \sigma_x^i(p) \mid e\}+1.\] Then $m\geq 1$, $\sigma_x^{m-1}(p)\mid e$ and $\sigma_x^m(p)\nmid e$. So $\sigma_x^m(p)\mid \sigma_x(e)$. If $\sigma_x^m(p) \nmid \det(F)$, then $F$ is an invertible matrix over $C[x]_{(\sigma_x^m(p))}$. Since $\sigma_x^{m-1}(p)$ is a factor of $e$, by the assumption we know that the greatest common divisor of all entries of $M$ and $e$ is $1$. So the greatest common divisor of all entries of $\sigma_x(M)$ and $\sigma_x(e)$ is also $1$. Replacing $M$ by $\sigma_x(M)$ in Lemma~\ref{Lem: gcd}, we obtain that the greatest common divisor of all entries of $\sigma_x(M)F$ (or $F\sigma_x(M)$) and $\sigma_x^m(p)$ is $1$. By~\eqref{EQ: weak compatibility}, we get $\sigma_x^m(p)$ must divide $e$ because $\sigma_x^m(p)\mid \sigma_x(e)$ and $G,M\in C[x]^{r\times r}$. This contradicts to the choice of $m$. Thus $\sigma_x^m(p) \mid\det(F)$.
	
	Let $n= \min\{ i \in \set N \mid \sigma_x^{-i}(p) \mid e\}$. Then $n\geq 0$, $\sigma_x^{-n}(p) \mid e$ and $\sigma_x^{-n-1}(p)\nmid e$. So~$\sigma_x^{-n}(p) \nmid \sigma_x(e)$. If $\sigma_x^{-n}(p)\nmid \det(G)$, then $G$ is an invertible matrix over $C[x]_{(\sigma_t^{-n}(p))}$. By the assumption, the greatest common divisor of all entries of $M$ and $\sigma_x^{-n}(p)$ is $1$. Replacing $F$ by $G$ in Lemma~\ref{Lem: gcd}, we obtain that the greatest common divisor of all entries of $GM$ and $\sigma_x^{-n}(p)$ is $1$. By~\eqref{EQ: weak compatibility}, we get $\sigma_x^{-n}(p)$ must divide $\sigma_x(e)$ because $F,M\in C[x]^{r\times r}$. This contradicts to the choice of $n$. Thus $\sigma_x^{-n}(p) \mid\det(G)$.	
	
	\eqref{it:division2} Suppose $p \mid \det(M)$ and $p\nmid e$. By~\eqref{EQ: weak compatibility}, we have
	\[\sigma_x\left(\frac{1}{e^r}\det(M)\right) \det(F) = \det(G)\frac{1}{e^r}\det(M).\]
	Since $\det(M)\neq 0$, it follows that
	\[\det(G)\sigma_x\left(\frac{1}{\det(M)}e^r\right) =\frac{1}{\det(M)}e^r\det(F).\]
	Now it comes back to the case~\eqref{it:division1} with the dimension of matrices being one (which is also the case in~\cite[Theorem 7]{AbramovPetkovsek01b}). So there exist $m,n\in \set N$, $m\geq 1$, $n\geq 0$, such that $\sigma_x^m(p)\mid \det(G)$ and $\sigma_x^{-n}(p)\mid \det(F)$.
\end{proof}

\begin{lemma}\label{Lem: compatiblity1}
	Let $e,e_t\in C[t,x]$ and $M, M_t\in C[t,x]^{r\times r}$ be invertible matrices over $C(t,x)$ satisfying the compatibility condition~\eqref{EQ: compatible}. Assume that the greatest common divisor (w.r.t. $t$ and $x$) of all entries of $M$ (resp. $M_t$) and $e$ (resp. $e_t$) is $1$. If $p\in C[t,x]$ is an irreducible factor of $e\det(M)$ and $\deg_t(p)\neq 0$, then there exist $i,j\in\set Z$, $i>0$, such that $\sigma_t^i\sigma_x^j(p)$ divides $e_t\det(M_t)$.
\end{lemma}
\begin{proof}
	Note that $p$ is also an irreducible polynomial in $C(x)[t]$. We rewrite~\eqref{EQ: compatible} as
	\[\sigma_t\left(\frac{1}{e}M\right)\underbrace{\sigma_x(e_t)M_t}_{=:F} = \underbrace{e_t\sigma_x(M_t)}_{=:G} \frac{1}{e}M.\]
	
	(\romannumeral 1) If $p \mid e$, then by Proposition~\ref{Prop: division}.\eqref{it:division1}, there exists $m\geq 1$ such that $\sigma_t^m(p)\mid \det(\sigma_x(e_t)M_t)$. Since $p$ is irreducible, it follows that $\sigma_t^m(p) \mid \sigma_x(e_t)$ or $\sigma_t^m(p)\mid \det(M_t)$. Take $(i,j)=(m, -1)$ in the former case, $(i,j) = (m,0)$ in the latter.
	
	(\romannumeral 2) If $p \mid \det(M)$ and $p\nmid e$, then by Proposition~\ref{Prop: division}.\eqref{it:division2}, there is $m\geq 1$ such that $\sigma_t^m(p) \mid \det(e_t\sigma_x(M_t))$. This implies that $\sigma_t^m(p) \mid e_t$ or $\sigma_t^m(p)\mid \sigma_x(\det(M_t))$. Take $(i,j) = (m,0)$ in the former case, $(i,j)=(m,-1)$ in the latter.
\end{proof}
\begin{lemma}\label{Lem: compatiblity2}
	Let $e,e_t\in C[t,x]$ and $M, M_t\in C[t,x]^{r\times r}$ be invertible matrices over $C(t,x)$ satisfying the compatibility condition~\eqref{EQ: compatible}. Assume that the greatest common divisor (w.r.t. $t$ and $x$) of all entries of $M$ (resp. $M_t$) and $e$ (resp. $e_t$) is $1$ and $\gcd(\det(M), \sigma_x^i(e))=1$ (w.r.t. $x$) for all $i\in \set Z\setminus\{0\}$. If $q\in C[t,x]$ is an irreducible factor of $e_t\det(M_t)$ and $\deg_x(q)\neq0$, then there exists $k\in\set Z$ such that $\sigma_x^k(q)$ divides $e\det(M)$.
\end{lemma}
\begin{proof}
	Note that $q$ is also an irreducible polynomial in $C(t)[x]$. We rewrite~\eqref{EQ: compatible} as
	\[\underbrace{e\sigma_t(M)}_{=:G} \frac{1}{e_t}M_t = \sigma_x\left(\frac{1}{e_t}M_t\right)\underbrace{\sigma_t(e)M}_{=:F}.\]
	
	(\romannumeral 1) If $q \mid e_t$, then by Proposition~\ref{Prop: division}.\eqref{it:division1}, there exist $m\geq 1$ and $n\geq 0$ such that 
	\[\sigma_x^m(q) \mid \det(\sigma_t(e)M)\quad\text{and}\quad \sigma_x^{-n}(q)\mid \det(e\sigma_t(M)).\]
	Since $q$ is irreducible, it follows that $\sigma_x^m(q)$ divides  $\sigma_t(e)$ or $\det(M)$, and $\sigma_x^{-n}(q)$ divides $e$ or $\sigma_t(\det(M))$. If $\sigma_x^m(q) \mid \sigma_t(e)$ and $\sigma_x^{-n}(q)\mid \sigma_t(\det(M))$, then $\sigma_t^{-1}\sigma_x^{-n}(q)$ would divide $\gcd (\det(M), \sigma_x^{-m-n}(e))$. If $\sigma_x^m(q) \mid \det(M)$ and $\sigma_x^{-n}(q)\mid e$, then $\sigma_x^{m}(q)$ would divide $\gcd (\det(M), \sigma_x^{m+n}(e))$. Since $m+n\neq 0$, these two cases can not happen by the assumption. So we have $\sigma_x^m(q)\mid \det(M)$ (and $\sigma_x^{-n}(q)\mid \det(\sigma_t(M))$), or $\sigma_x^{-n}(q)\mid e$ (and $\sigma_x^m(q)\mid \sigma_t(e)$). Take $k=m$ in the former case, $k=-n$ in the latter.
	
	(\romannumeral 2) If $q\mid \det(M_t)$ and $q\nmid e_t$, then by Proposition~\ref{Prop: division}.\eqref{it:division2}, there exist $m\geq1$ and $n\geq 0$ such that
	\[\sigma_x^m(q) \mid \det(e\sigma_t(M))\quad\text{and}\quad \sigma_x^{-n} \mid \det(\sigma_t(e)M).\]
	Similar to~(\romannumeral1), since $\gcd(\det(M), \sigma_x^i(e))=1$ for all $i\in \set Z\setminus\{0\}$, it follows that $\sigma_x^m(q)\mid e$ or $\sigma_x^{-n}(q)\mid \det(M)$. Take $k=m$ in the former case, $k=-n$ in the latter.
\end{proof}

\begin{proof}[Proof of Theorem~\ref{Thm: compatibility} (continued)]
	Note that polynomials in $C[x]$ and $C[t]$ are integer-linear. 
	
	If $p$ is an irreducible factor of $e\det(M)$ with $\deg_t(p)\neq0$ and $\deg_x(p)\neq0$, then by Lemma~\ref{Lem: compatiblity1}, there exist $i,j\in\set Z$, $i>0$, such that $\sigma_t^i\sigma_x^j(p)$ divides $e_t\det(M_t)$. By Lemma~\ref{Lem: compatiblity2}, there exists $k\in\set Z$ such that $\sigma_t^i\sigma_x^{j+k}(p)$ divides $e\det(M)$. Hence by Lemma~\ref{Lem: integer-linear thm6}, $e\det(M)$ is integer linear.
	
	If $q$ is an irreducible factor of $e_t\det(M_t)$ with $\deg_t(q)\neq0$ and $\deg_x(q)\neq0$, then by Lemma~\ref{Lem: compatiblity2}, there exists $k\in\set Z$ such that $\sigma_x^k(q)$ divides $e\det(M)$. By Lemma~\ref{Lem: compatiblity1}, there exist $i,j\in\set Z$, $i>0$, such that $\sigma_t^i\sigma_x^{j+k}(q)$ divides $e_t\det(M_t)$. Hence by Lemma~\ref{Lem: integer-linear thm6}, $e_t\det(M_t)$ is integer-linear.
	
	Since $e,\det(M),e_t,\det(M_t)$ are nonzero, it follows that $e,\det(M), e_t, \det(M_t)$ are integer-linear.
\end{proof}
\section{Properness implies existence of telescopers}\label{sec: proper implies existence}
For a proper hypergeometric term, Huang~\cite{Huang16} gave lower and upper bounds on the order of its minimal telescopers. Using the notations in Section~\ref{sec:integer-linear}, we shall give an upper bound on the order of minimal telescopers for proper P-recursive sequences, which provides a proof of the existence of telescopers for such sequences.

\begin{lemma}\label{Lem: tele reduction}
	Let $W$ be a suitable basis of $A\cong K(x)[S_x]/\<L>$ at $\{\beta_1,\ldots,\beta_I\}$. Let $e,\tilde e\in C[t,x]$ and $M,\tilde M\in C[t,x]^{r\times r}$ be such that $S_xW=\frac{1}{e}MW$ and $\frac{1}{\tilde e}\tilde M=(\frac{1}{e}M)^{-1}$. Let $q=h(mt+nx)\in C[t,x]$ be a monic irreducible integer-linear polynomial, where $h\in C[z]$, $m,n\in\set Z$ with $n>0$ and $\gcd(m,n)=1$. Let $p_i\in C[t,x]$ be the monic minimal polynomial of $\beta_i$ over $K$. 
	For any $a\in K[x]^r$ and $\xi\in \set N[\tau,\tau^{-1}]$ with $\tau =\tau^{(m,n)}$, 
	\begin{enumerate}[(i)]
		\item\label{it:tele-red1} if $\gcd(q, \sigma_t^\ell\sigma_x^j(p_i))=1$ for all $i\in\{1,\ldots, I\}$ and $\ell,j\in \set Z$, then there exist $g\in A$ and $c\in K[x]^r$ such that
		\begin{equation}\label{EQ: tele-red}
			\frac{aW}{q^\xi}=\Delta_x(g) +  \frac{cW}{b},
		\end{equation}
		where $b=e\tilde e\sum_{j=0}^{n-1}\sigma_t^j(q^{||\xi||_\tau})$.
		\item\label{it:tele-red2} if $q=p_i$ for some $i\in\{1,\ldots, I\}$, then there exist $g\in A$ and $c\in K[x]^r$ such that~\eqref{EQ: tele-red} holds with $b =e\tilde e\sum_{p\in \Lambda_q}\sum_{j=0}^{n-1}\sigma_t^j(p^{||\xi||_\tau})$, where $\Lambda_q$ is the set consisting of all polynomials in $\{p_1,\ldots,p_I\}$ that are $\<\sigma_t,\sigma_x>$-shift equivalent to $q$.
	\end{enumerate}
\end{lemma}
\begin{proof}
	We write $\xi=\sum_{i=\ell}^\rho k_i \tau^i$, where $\ell, \rho\in \set Z$ with $\ell\leq \rho$, $k_i\in \set N$ and only finitely many $k_i$ are nonzero. Note that the irreducible polynomials $\tau^{i}(q)$ with $i\in\set Z$ are pairwise coprime. Applying the partial fraction decomposition
	of rational functions to all coeﬀicients of $\frac{1}{q^\xi}aW$, we get
	\[\frac{aW}{q^\xi} = \left(p + \sum_{i=\ell}^\rho \frac{a_i}{q^{k_i\tau^i}}\right)W,\]
	where $p, a_i\in K[x]$, $\deg_x(a_i)<\deg_x(q_i^{k_i})$. Since $||k_i\tau^i||_\tau = k_i \leq ||\xi||_\tau$ for all $i$ with $\ell \leq i\leq \rho$, we only need to show the case that $\xi = k\tau^i$ with $k\in \set N$ and $i\in \set Z$. By~\eqref{EQ:tau p}, there exist $i_0,j_0\in\set Z$ with $0\leq i_0 \leq n-1$ such that $\tau^i(q) = \sigma_t^{i_0}\sigma_x^{j_0}(q)$.
	
	\eqref{it:tele-red1} By the assumption, $\gcd(\sigma_t^{i_0}(q), \sigma_x^j(p_i))=1$ for all $i\in\{1,\ldots, I\}$ and $j\in \set Z$. So by Lemma~\ref{Lem: AP step}.\eqref{it:AP-red1}, there exist $g\in A$ and $\tilde c\in K[x]^r$ such that 
	\[\frac{aW}{q^{k\tau^i}} = \frac{aW}{\sigma_t^{i_0}\sigma_x^{j_0}(q^k)}=\Delta_x(g) + \frac{\tilde cW}{\sigma_t^{i_0}(q^k) e\tilde e}=\Delta_x(g)+\frac{cW}{b},\]
	where $c = \tilde c b/(e\tilde e \sigma_t^{i_0}(q^k))$. Then $c\in K[x]^r$, because $e\tilde e\sigma_t^{i_0}(q^k)$ is a factor of $b$.
	
	\eqref{it:tele-red2} Note that $q\in \Lambda_q$. If $\gcd(\sigma_t^{i_0}(q), \sigma_x^j(p_{i'}))=1$ for all $i'\in\{1,\ldots, p_I\}$ and $j\in \set Z$, then~\eqref{EQ: tele-red} holds by the same argument as in~\eqref{it:tele-red1}. If $
	\gcd(\sigma_t^{i_0}(q),\sigma_x^{j}(p_{i'}))\neq 1$ for some $i'\in\{1,\ldots, I\}$ and $j\in \set Z$, then $\sigma_t^{i_0}(q)=\sigma_x^{j}(p_{i'})$ because $q$ and $p_{i'}$ are monic irreducible polynomials. So $p_{i'}\in \Lambda_q$. By Lemma~\ref{Lem: AP step}.\eqref{it:AP-red2}, there exist $g\in A$ and $\tilde c\in K[x]^r$ such that
	\[\frac{aW}{q^{k\tau^i}}=\frac{aW}{\sigma_t^{i_0}\sigma_x^{j_0}(q^k)}=\frac{aW}{\sigma_x^{j+j_0}(p_{i'}^k)} = \Delta_x(g) + \frac{\tilde c W}{ p_{i'}^ke\tilde e}=\Delta_x(g) + \frac{cW}{b},\]
	where $c=\tilde cb/(e\tilde ep_{i'}^k)\in K[x]^r$.
\end{proof}

\begin{theorem}\label{Thm: tele add}
	With the notations introduced in Theorem~\ref{Thm:add decomp},  we assume that $f\in A$ has a decomposition
	\begin{equation}\label{EQ:AP-step tele}
		f = \Delta_x(g) + \frac{1}{d} PW + \frac{1}{e} RW,
	\end{equation}
	where $g\in A$, $d\in K[x]$, $P,R\in K[x]^r$ with $\deg_x(P)<\deg_x(d)$ and $d\prod_{i=1}^I(x-\beta_i)$ being shift-free with respect to $x$. If $d$ is integer-linear, then there exists a polynomial $b\in K[x]$ such that for every nonnegative integer $\ell\geq 0$, $S_t^\ell f$ can be decomposed into
	\begin{equation}\label{EQ: tele add}S_t^\ell f = \Delta_x(\tilde g_\ell)+\frac{1}{b}P_\ell W + \frac{1}{a}Q_\ell V,
	\end{equation}
	where $\tilde g_\ell\in A$, $P_\ell\in K[x]^r$ with $\deg_x(R_\ell)<\deg_x(b)$ and  $Q_\ell\in N_V$.
\end{theorem}
\begin{proof}
	By the assumption, $W$ is a suitable basis of~$A$ at $\{\beta_1,\ldots, \beta_I\}$ that is normal at infinity. Let $e_t\in C[t,x]$ and $M_t\in C[t,x]^r$ be such that $S_tW=\frac{1}{e_t}M_tW$. Then for any integer $\ell\geq 0$, by~\eqref{EQ:AP-step tele},
	\begin{equation}\label{EQ: St f expansion}
		S_t^\ell f = \Delta_x(S_t^\ell g )+\left(\frac{1}{\sigma_t^\ell(d)}\sigma_t^\ell(P)+\frac{1}{\sigma_t^\ell(e)}\sigma_t^\ell(R)\right)\prod_{j=0}^{\ell-1}\sigma_t^j \left(\frac{1}{e_t}M_t\right)W.
	\end{equation}
	By Theorem~\ref{Thm: compatibility}, $e$ and $e_t$ are integer-linear. For simplicity, we may assume that $d,e,e_t$ are monic. Since $d$ is also integer-linear, by~\eqref{EQ: shiftless decomp} we can write
	\begin{equation}\label{EQ: three integer reps}
		d = \prod_{i=1}^{s} q_i^{\theta_i},\quad e =\prod_{i=1}^s q_i^{\eta_i}\quad \text{and}\quad e_t = \prod_{i=1}^s q_i^{\xi_i},
	\end{equation}
	where $s\in \set N$, each $q_i=h_i(m_it + n_ix) \in C[t,x]$ is a monic irreducible integer-linear polynomial with $h_i\in C[z]$, $m_i,n_i\in \set Z$, $n_i>0$, $\gcd(m_i,n_i)=1$, and $\theta_i$, $\eta_i$, $\xi_i\in \set N[\tau_i,\tau_i^{-1}]$ with $\tau_i=\tau^{(m_i,n_i)}$. Moreover,  $q_i$ and $q_{j}$ are not $\<\sigma_t,\sigma_x>$-shift equivalent whenever $i\neq j$. Let~$k= \max_{i=1}^{s}\left\{\max\{||\theta_i||_{\tau_i}, ||\eta_i||_{\tau_i}\}+||\xi_i||_{\tau_i}^*\right\}$. Let $p_i\in C[t,x]$ be the monic minimal polynomial of $\beta_i$ over~$K$. If $q_i$ is $\<\sigma_t,\sigma_x>$-shift equivalent to $p_{i'}$ for some $i'\in\{1,\ldots, I\}$, we may choose $q_i=p_{i'}$ as a representative in this shift equivalence class. For a polynomial $q\in\{q_1,\ldots, q_s\}$, let $\Lambda_q$ be the set consisting of all polynomials in $\{p_1,\ldots,p_I\}$ that are $\<\sigma_t,\sigma_x>$-shift equivalent to $q$. Let $b\in C[t,x]$ be the least common multiple of $b_1$ and $b_2$, where $b_1 = e\tilde e\prod_{i=1}^s \prod_{j=0}^{n_i-1}\sigma_t^j(q_i^k)$ and $b_2 =e\tilde e \sum_{i=1}^s\prod_{p\in \Lambda_{q_i}}\prod_{j=0}^{n_i-1}\sigma_t^j(p^k)$. 
	
	We claim that for every integer $\ell \geq 0$, 
	there exist $g_\ell\in A$ and $P_\ell, R_\ell\in K[x]^r$ with $\deg_x(P_\ell) <\deg_x(b)$ such that
	\begin{equation}\label{EQ: add W}
		S_t^\ell f = \Delta_x(g_\ell)+\left(\frac{1}{b}P_\ell+\frac{1}{e}R_\ell\right)W.
	\end{equation}
	If the claim is true, then by Theorem~\ref{Thm:add decomp}, there exists $Q_\ell \subset N_V$ such that~\eqref{EQ: tele add} holds for some $\tilde g_\ell \in A$. For simplicity, if $S_t^\ell f- h_\ell \in \Delta_x(A)$ with $h_\ell\in A$, we write $S_t^\ell f \equiv h_\ell \mod \Delta_x(A)$.
	Substituting~\eqref{EQ: three integer reps} into~\eqref{EQ: St f expansion} yields that
	\begin{align*}
		S_t^\ell f \equiv \left(\frac{1}{\sigma_t^\ell(\prod_{i=1}^s q_i^{\theta_i})}\sigma_t^{\ell }(P)+\frac{1}{\sigma_t^\ell(\prod_{i=1}^s q_i^{\eta_i})}\sigma_t^\ell(R)\right)\prod_{j=0}^{\ell-1} \frac{1}{\sigma_t^j(\prod_{i=1}^sq_i^{\xi_i})} \sigma_t^j(M_t)W \mod \Delta_x(A).
	\end{align*}
	Using the notations in Section~\ref{sec:integer-linear}, we have $\sigma_t(q_i)=\tau_i^{n_i}(q_i)$ and hence $\sigma_t^{\ell} (q_i^{\zeta_i}) = \tau_i^{\ell n_i}(q_i^{\zeta_i}) = q_i^{\tau_i^{\ell n_i}\zeta_i}$ for every  $\zeta_i\in \set N[\tau_i, \tau_i^{-1}]$. Since $\sigma_t$ is a ring homomorphism, it follows that
	\begin{align*}
		S_t^{\ell}f &\equiv  \left(\frac{1}{\prod_{i=1}^s \sigma_t^\ell(q_i^{\theta_i})}\sigma_t^{\ell }(P)+\frac{1}{\prod_{i=1}^s\sigma_t^\ell( q_i^{\eta_i})}\sigma_t^\ell(R)\right)\prod_{j=0}^{\ell-1} \frac{1}{\prod_{i=1}^s\sigma_t^j(q_i^{\xi_i})} \sigma_t^j(M_t)W \mod \Delta_x(A)\\ &=
		\left(\frac{1}{\prod_{i=1}^s  q_i^{\tau_i^{\ell n_i}\theta_i}}\sigma_t^{\ell }(P)+\frac{1}{\prod_{i=1}^s q_i^{\tau_i^{\ell n_i}\eta_i}}\sigma_t^\ell(R)\right)\prod_{j=0}^{\ell-1} \frac{1}{\prod_{i=1}^sq_i^{\tau_i^{jn_i}\xi_i}} \sigma_t^j(M_t)W\\
		&= \underbrace{\frac{1}{\prod_{i=1}^s q_i^{\tau_i^{\ell n_i}\theta_i + \sum_{j=0}^{\ell - 1} \tau_i^{jn_i}\xi_i}} \sigma_t^\ell(P) \prod_{j=0}^{\ell-1} \sigma_t^j(M_t)W}_{=:f_\ell^{(1)}}+ \underbrace{\frac{1}{\prod_{i=1}^s q_i^{\tau_i^{\ell n_i}\eta_i + \sum_{j=0}^{\ell - 1} \tau_i^{jn_i}\xi_i}} \sigma_t^\ell(R) \prod_{j=0}^{\ell-1} \sigma_t^j(M_t)W}_{=:f_\ell^{(2)}}.
	\end{align*}
	For each $i\in\{1,\ldots,s\}$, by Lemma~\ref{Lem: norm}, 
	\begin{center}
		$\left|\left| \tau_i^{\ell n_i}\theta_i +\sum_{j=0}^{\ell-1}\tau_i^{jn_i}\xi_i\right|\right|_{\tau_i} \leq \left|\left|\tau_i^{\ell n_i}\theta_i \right|\right|_{\tau_i} + \left|\left|\sum_{j=0}^{\ell-1}\tau_i^{jn_i}\xi_i\right|\right|_{\tau_i}\leq \left|\left| \theta_i\right|\right|_{\tau_i}+ \left| \left| \xi_i\right|\right|^*_{\tau_i}\leq k$
	\end{center}
	and similarly
	\begin{center}
		$\left|\left| \tau_i^{\ell n_i}\eta_i +\sum_{j=0}^{\ell-1}\tau_i^{jn_i}\xi_i\right|\right|_{\tau_i} \leq \left|\left| \eta_i\right|\right|_{\tau_i}+ \left| \left| \xi_i\right|\right|^*_{\tau_i}\leq k$.
	\end{center}
	For each $j\in\{1,2\}$, since the $q_i$'s are in distinct $\<\sigma_t,\sigma_x>$-orbits, by the partial fraction decomposition of rational functions and Lemma~\ref{Lem: tele reduction}, there exists $c_\ell^{(j)} \in K[x]^r$ such that
	\begin{equation}\label{EQ: bound f_i}
		f_\ell^{(j)}\equiv \frac{c_\ell^{(j)}W}{b} \mod \Delta_x(A).
	\end{equation}
	By Corollary~\ref{Cor: AP-remainder}, the right hand side of~\eqref{EQ: bound f_i} can be rewritten as the sum of two parts and we get
	\[f_\ell^{(j)} \equiv \frac{1}{b} P_\ell^{(j)}W +\frac{1}{e}R_\ell^{(j)}W\mod \Delta_x(A),\]
	where $P_\ell^{(j)}, R_\ell^{(j)} \in K[x]^r$ and $\deg_x(P_\ell^{(j)}) < \deg_x(b)$. Then we obtain the decomposition~\eqref{EQ: add W} by setting $P_\ell = P_\ell^{(1)}+P_\ell^{(2)}$ and $R_\ell = R_\ell^{(1)}+R_\ell^{(2)}$. 
\end{proof}
\begin{corollary}\label{Cor: tele bound}
	With the assumptions and notations introduced in Theorem~\ref{Thm: tele add}, $f$ has a telescoper $T$ in  $C(t)[S_t]$ of order at most $r (\deg_x(b) + \dim_K(N_V))$.
\end{corollary}
\begin{proof}
	Let $T=\sum_{i=0}^\rho c_\ell S_t^\ell$ with $\rho\in\set N$ and $c_\ell\in C(t)$. Then by~\eqref{EQ: tele add},
	\[Tf = \sum_{\ell=0}^\rho c_\ell S_t^\ell f = \Delta_x\left(\sum_{\ell=0}^\rho c_\ell S_t^\ell\tilde g_\ell\right) + \frac{1}{b} \left(\sum_{\ell=0}^\rho c_\ell P_\ell\right) W + \frac{1}{a} \left(\sum_{\ell = 0}^\rho c_\ell Q_\ell\right) V.\]
	Then $\deg_x(\sum_{\ell=0}^\rho c_\ell P_\ell)<\deg_x(b)$ and $\sum_{\ell =0 }^\rho c_\ell Q_\ell \in N_V$. The set of all polynomials in $K[x]$ of degree less than $\deg_x(b)$ forms a $K$-vector space of dimension $\deg_x(b)$. The space $N_V$ is also a $K$-vector space of finite dimension. So \[\sum_{\ell=0}^\rho c_\ell P_\ell =0 \quad \text{and}\quad \sum_{\ell=0}^\rho c_\ell Q_\ell =0\]
	form a linear system over $K$ with $\rho + 1$ unknowns $c_0,\ldots, c_\rho$ and at most  $r(\deg_x(b)+\dim_K(N_V))$ equations. This linear system has nontrivial solutions whenever $\rho \geq r(\deg_x(b)+\dim_K(N_V))$. This implies that $f$ has a telescoper $T$ of order at most $r(\deg_x(b)+\dim_K(N_V))$.
\end{proof}

\begin{proof}[Proof of Theorem~\ref{Thm: main}]
	Suppose $f=\Delta_x(g) + h$ with $h$ being proper. By Definition~\ref{Defi: proper}, there exists a suitable basis $W$ of $A$ with respect to $x$ at $\{\beta_1,\ldots,\beta_I\}$ such that the denominator $\tilde u$ of $h$ with respect to $W$ is integer-linear. Let $e, \tilde e\in C[t,x]$ and $M,\tilde M\in C[t,x]^{r\times r}$ be such that $S_xW=\frac{1}{e}MW$ and $\frac{1}{\tilde e}\tilde M = (\frac{1}{e}M)^{-1}$. By Corollary~\ref{Cor: AP-remainder}, $f$ has a decomposition \[f=\Delta_x(\tilde g) + \frac{1}{d}PW+\frac{1}{e}RW\]
	where $\tilde g\in A$, $d\in K[x]$, $P,R\in K[x]^r$ with $\deg_x(P)<\deg_x(d)$ and $d\prod_{i=1}^I(x-\beta_i)$ being shift-free with respect to~$x$. As described in the generalized Abramov-Petkov\v sek reduction, see~\cite[Section 4]{chen23b}, every irreducible factor of $d$ is $\<\sigma_x>$-equivalent to some irreducible factor of $e\tilde e\tilde u$. By Theorem~\ref{Thm: compatibility}, $e$ and $\det(M)$ are integer-linear. By Corollary~\ref{Cor: tilde e integer-linear}, $\tilde e$ divides $\det(M)$ and hence $\tilde e$ is integer-linear. Therefore $d$ is integer-linear. By Theorem~\ref{Thm: tele add} and Corollary~\ref{Cor: tele bound}, $f$ has a telescoper of type $(S_t;S_x)$. The proof of the other direction of Theorem~\ref{Thm: main} will be given in Section~\ref{sec: existence implies proper}.
\end{proof}

\section{Stems of P-recursive sequences}\label{sec: stem}
The notion of stems was introduced by Abramov~\cite{Abramov03} in deciding the existence of telescopers for hypergeometric terms. It gives an equivalent description of proper hypergeometric terms. Now we extend this notion to P-recursive sequences.

\begin{defi}
	Let $Q = \frac{a}{u}$ with $a=(a_1,\ldots, a_r)\in C[t,x]^r$, $u\in C[t,x]$ and $\gcd(a_1,\ldots, a_r, u) = 1$. The {\em stem} of $Q$ is defined as the following product
	\[\bar u := \prod\{ p \in C[t,x] \mid p \text{ is a monic irreducible factor of $u$ and not integer-linear}\}.\]
	Then $u=\bar u u_\infty$, where $\bar u$ has no integer-linear factors and  $u_\infty\in C[t,x]$ is integer-linear.
\end{defi}

\begin{thm}\label{Thm: indep of bases}
Let $W_1$ and $W_2$ be two bases of $A$. Let $e_1,e_2\in C[t,x]$ and $M_1,M_2\in C[t,x]^{r\times r}$ be such that $S_xW_1=\frac{1}{e_1}M_1W_1$ and $S_x W_2=\frac{1}{e_2}M_2 W_2$. Let $f\in A$ and we write $f = Q_1W_1 = Q_2W_2$, where $Q_1,Q_2\in C(t,x)^r$. If $e_1,e_2, \det(M_1), \det(M_2)$ are integer-linear, then $Q_1$ and $Q_2$ have the same stem.
\end{thm}
\begin{proof}	
	Let $v\in C[t,x]$ and $T\in C[t,x]^{r\times r}$ be such that $W_2 = \frac{1}{v}TW_1$ and the greatest common divisor of all entries of $T$ and $v$ is $1$. Then $S_xW_2 = \sigma_x(\frac{1}{v}T)S_xW_1=\sigma_x(\frac{1}{v}T)\frac{1}{e_1}M_1(\frac{1}{v}T)^{-1}W_2=\frac{1}{e_2}M_2W$. Since $W_2$ is a basis of $A$, equating the coefficients of $W_2$ and multiplying $e_1 e_2$ yield that
	\[\sigma_x\left(\frac{1}{v}T\right)\underbrace{e_2 M_1}_{=:F} = \underbrace{e_1M_2}_{=:G}\frac{1}{v}T.\]
	
	We claim that $v$ is integer-linear. If $p$ is an irreducible factor of $v$ with $\deg_x(p)\neq 0$, then by Proposition~\ref{Prop: division}.\eqref{it:division1}, there exists $m\geq 1$ such that $\sigma_x^m(p)$ divides $\det(e_2 M_1)$. So $p$ is integer-linear because $\det(e_2M_1)=e_2^r\det(M_1)$ is integer-linear. 
	Thus the claim is true. 
	
	Since $W_2=\frac{1}{v}TW_1$, we have $Q_1= Q_2\frac{1}{v}T$. Since $T\in C[t,x]^r$ and $v$ is integer-linear, every non integer-linear irreducible factor of the denominator of $Q_1$ comes from the denominator of $Q_2$. So the stem of $Q_1$ divides the stem of $Q_2$. Similarly,
	we can write $W_1=\frac{1}{\tilde v}\tilde TW_2$, where $\tilde T\in C[t,x]^{r\times r}$ and $\tilde v\in C[t,x]$ is integer-linear. So the stem of $Q_2$ divides the stem of $Q_1$. Therefore $Q_1$ and $Q_2$ have the same stem because the stem of a vector is a monic polynomial. 
\end{proof}

\begin{corollary}\label{Cor: indep of suitable}
	Let $W_1$ and $W_2$ be two suitable bases of $A$ with respect to $x$. Let $f\in A$ and we write $f = \frac{1}{u_1}a_1W_1 = \frac{1}{u_2}a_2W_2$, where $a_1,a_2\in C[t,x]^r$, $u_1,u_2\in C[t,x]$ and the greatest common divisor of all entries of $a_1$ (resp. $a_2$) and $u_1$ (resp. $u_2$) is $1$. Then 
	\begin{enumerate}[(i)]
		\item\label{it:stem1} $\frac{a_1}{u_1}$ and $\frac{a_2}{u_2}$ have the same stem;
		\item $u_1$ is integer-linear if and only if $u_2$ is integer-linear.
	\end{enumerate}
\end{corollary}
\begin{proof}
	By the definition of stems, $u_1$ is integer-linear if and only if the stem of $\frac{a_1}{u_1}$ is $1$. So we only need to prove the item~\eqref{it:stem1}. Let $e_1,e_2\in C[t,x]$ and $M_1,M_2\in C[t,x]^{r\times r}$ be such that $S_xW_1=\frac{1}{e_1}M_1W_1$ and $S_xW_2=\frac{1}{e_2}M_2 W_2$. Since $W_1$ and $W_2$ are suitable bases, it follows from Theorem~\ref{Thm: compatibility}.\eqref{it:compat2} that $e_1,e_2, \det(M_1), \det(M_2)$ are integer-linear. The conclusion follows from Theorem~\ref{Thm: indep of bases}.
\end{proof}

Let $W$ be a suitable basis of $A$ with respect to $x$. Let $f=QW\in A$, where $Q\in C(t,x)^r$. The stem $\bar u$ of $Q$ is called the {\em stem} of $f$ with respect to $x$. Then $f$ can be written as
\begin{equation}\label{EQ: stem}
	f = \frac{aW}{\bar u u_\infty},
\end{equation}
where $a=(a_1,\ldots, a_r)\in C[t,x]^r$, $\bar u\in C[t,x]$ has no integer-linear factors, $u_\infty\in C[t,x]$ is integer-linear and  $\gcd(a_1,\ldots,a_r, \bar u u_\infty) = 1$. 
By Corollary~\ref{Cor: indep of suitable}, the stem of $f$ is uniquely defined and independent of the choice of suitable bases. By the definition of properness and stems, we have the following lemma.
\begin{lemma}\label{Lem: stem 1}
	Let $f\in A$. Then $f$ is proper with respect to $x$ if and only if the stem of $f$ with respect to $x$ is equal to $1$.
\end{lemma}

\begin{corollary}\label{Cor: sym}
	Let $f\in A$. Then the stem of $f$ with respect to $x$ is equal to the stem of $f$ with respect to~$t$. In particular, $f$ is proper with respect to $x$ if and only if $f$ is proper with respect to $t$.
\end{corollary}
\begin{proof}
	Let $W_1$ be a suitable basis with respect to $x$. Let $e_1\in C[t,x]$ and $M_1\in C[t,x]^{r\times r}$ be such that $S_xW_1=\frac{1}{e_1}M_1W$. By Theorem~\ref{Thm: compatibility}.\eqref{it:compat2}, $e_1$ and $
	\det(M_1)$ are integer-linear.
	
	
	Let $W_2$ be a suitable basis of $A$ with respect to $t$. By Assumption~\ref{Assump}, such a basis $W_2$ exists and it can be computed similarly to a suitable basis with respect to $x$. Let $e_2\in C[t,x]$ and $M_2\in C[t,x]^{r\times r}$ be such that $S_xW_2=\frac{1}{e_2}M_2W$.  Swapping $x$ and $t$ in Theorem~\ref{Thm: compatibility}.\eqref{it:compat2}, we obtain that $e_2$ and $
	\det(M_2)$ are integer-linear. 
	
	We write $f=Q_1W_1=Q_2W_2$, where $Q_1,Q_2\in C[t,x]^r$. By Theorem~\ref{Thm: indep of bases}, $Q_1$ and $Q_2$ have the same stem. So the stem of $f$ with respect to $x$ is equal to the stem of $f$ with respect to $t$. Then the second assertion follows from Lemma~\ref{Lem: stem 1}.
\end{proof}

\begin{corollary}\label{Cor: hypergeo}
	Let $J\subseteq C(t,x)[S_t,S_x]$ be the annihilating ideal of a hypergeometric term $F$. Then $F$ is a proper hypergeometric term if and only if $1\in A = C(t,x)[S_t,S_x]/J$ is proper with respect to both $t$ and~$x$.
\end{corollary}
\begin{proof}
	Since $F$ is a hypergeometric term, its annihilating ideal $J$ is generated by $S_t - u_t$ and $S_x - u_x$ for some $u_t, u_x\in C(t,x)$. Then $A$ is $C(t,x)$-vector space of dimension one. The hypergeometric term $F$ corresponds to $f=1\in A$. Let $W$ be a suitable basis of $A$ with respect to $x$. By Theorem~\ref{Thm: compatibility}, there are integer-linear polynomials $e,M\in C[t,x]$ such that $S_xW=\frac{1}{e}MW$ and $\gcd(M,e)=1$. Let $Q\in C(t,x)$ be such that $f=QW$. Since $f$ is a basis of $A$, the rational function $Q$ is nonzero. Then
	\[S_xf = \sigma_x(Q)S_xW=\sigma_x(Q)\frac{1}{e}MW=\sigma_x(Q)\frac{1}{e}MQ^{-1}f= u_xf.\]
	By the last equality, we get \begin{equation}\label{EQ: RNF}
		u_x = \frac{M}{e}\frac{\sigma_x(Q)}{Q}.
	\end{equation}
	Since $\gcd(M,e)=1$, it follows from Theorem~\ref{Thm: suitable}.\eqref{it:suitable4} that $\gcd(M,\sigma_x^i(e))=1$ for all $i\in \set Z$. So~\eqref{EQ: RNF} is a rational normal form (see its definition in~\cite[Section 2.2]{Abramov03}) of $u_x$ with respect to $x$. By Corollaries 1 and 2 in~\cite{Abramov03}, $F$ is a proper hypergeometric term if and only the stem of $Q$ is equal to $1$. By Lemma~\ref{Lem: stem 1} and Corollary~\ref{Cor: sym}, the latter is equivalent to the fact that $f=QW$ is proper with respect to both $t$ and $x$.
\end{proof}	

The notion of proper $\partial$-finite ideals was first introduced by Chen, Kauers and Koutschan~\cite{chen14}, which provides a sufficient condition for the existence of telescopers for $1\in A$. The properness in Definition~\ref{Defi: proper} and~\cite[Definition 4]{chen14} are connected as follows.
\begin{remark}\label{Rem: compare proper}
Let $A=C(t,x)[S_t,S_x]/J$ be a finite-dimensional vector space over $C(t,x)$ satisfying Assumption~\ref{Assump}. For a basis $W$ of $A$, let $e\in C[t,x]$ and $M\in C[t,x]^{r\times r} $ be such that $S_xW=\frac{1}{e}MW$.
The ideal $J$ is proper with respect to $x$ in the sense of~\cite[Definition 4]{chen14} if and only if there exists an {\em $x$-admissible} basis $W$ of $A$, i.e., the element $1\in A$ is represented as $1=aW$ with $a\in C(t)[x]^r$, for which the polynomial $e$ is integer-linear. 

Suppose that $1\in A$ is proper in the sense of Definition~\ref{Defi: proper}. Then there exists a suitable basis $W$ with respect to $x$ such that $1=\frac{a}{u}W$, where $a\in C[t,x]^r$ and $u\in C[t,x]$ is integer-linear. By Theorem~\ref{Thm: suitable}, for this basis $W$, the polynomial $e$ is integer-linear. Let $U=\frac{1}{u}W$. Then $U$ is an $x$-admissible basis because $1=aU$. Since $S_xU = \frac{1}{\sigma_x(u)}S_xW=\frac{1}{\sigma_x(u)}\frac{1}{e}MW =\frac{u}{\sigma_x(u)}\frac{1}{e}MU$ with $\sigma_x(u)e$ being integer-linear, it follows that $J$ is proper in the sense of~\cite[Definition 4]{chen14}. 
\end{remark}

\section{Existence of telescopers implies properness}\label{sec: existence implies proper}
A polynomial $u\in K[x]$ is {\em spread} with respect to $x$ if for any irreducible polynomial $p$ which divides $u$ there is $m\in \set Z\setminus\{0\}$ such that $\sigma_x^m(p)$ divides $u$. Spread polynomials were introduced in~\cite{Abramov03} to give the existence criterion of telescopers for hypergeometric terms. 
Following the underlying idea in Abramov's paper~\cite{Abramov03}, we shall prove Theorem~\ref{Thm: main_shiftfree} using the generalized Abramov-Petkov\v sek reduction in Section~\ref{sec: add}.


\begin{theorem}\label{Thm:spread_sum}
	Let $f\in A$ be a bivariate P-recursive sequence whose stem is not spread with respect to~$x$. Then $f$ is not summable in $A$ with respect to $x$.
\end{theorem}

\begin{proof}
	Let $W$ be a suitable basis of $A\cong K(x)[S_x]/\<L>$ at $\{\beta_1,\ldots,\beta_I\}$. Let $e\in C[t,x]$ and $M\in C[t,x]^{r\times r}$ be such that $S_x W =\frac{1}{e}MW$. By Corollary~\ref{Cor: AP-remainder}, $f$ can be decomposed into
	\begin{equation}\label{EQ: AP-decomp}
		f = \Delta_x(g)+ \frac{1}{d}RW+\frac{1}{e}SW,
	\end{equation}
	where $g\in A$, $d\in K[x]$, $R,S\in K[x]^r$ with $\deg_x(R)<\deg_x(d)$ and $d\prod_{i=1}^I(x-\beta_i)$ being shift-free.
	
	Let $\bar u\in C[t,x]$ be the stem of $f$. By~\eqref{EQ: stem}, $f$ can be written as $f =\frac{1}{\bar u u_{\infty}}aW$, where $a\in C[t,x]^r$ and $u_\infty\in C[t,x]$ is integer-linear. Since $\bar u$ is not spread, $\bar u$ has an irreducible factor $p$ such that $\sigma_x^m(p)\nmid \bar u$ for any $m\in\set Z\setminus\{0\}$.
	As shown in~\cite[Section 4]{chen23b}, there exist a decomposition~\eqref{EQ: AP-decomp} and an integer $m_0\in \set Z$ such that $\sigma_x^{m_0}(p) \mid de$.
	
	By the definition of stems, $p$ is not integer-linear. By Theorem~\ref{Thm: compatibility}.\eqref{it:compat2}, the polynomial $e$ is integer-linear. So $\sigma_x^{m_0}(p) \nmid e$ and hence $\sigma_x^{m_0}(p)\mid d$. Then $d$ is not a constant in $K$ because $p$ is of positive degree in $x$. By Corollary~\ref{Cor: AP-remainder}, $f$ is not summable in $A$ with respect to $x$.
\end{proof}

\begin{theorem}\label{Thm:spread_Tf}
	Let $f\in A$ be a bivariate P-recursive sequence which is not proper. Let the stem of $f$ be shift-free with respect to $x$. Then for any nonzero operator $T\in C(t)[S_t]$, the stem of $Tf$ is is not spread with respect to $x$.
\end{theorem}
\begin{proof}
	Let $\bar u\in C[t,x]$ be the stem of $f$. Let $W$ be a suitable basis of $A$ with respect to $x$. By~\eqref{EQ: stem}, $f$ can be written as $f =\frac{1}{\bar u u_{\infty}}aW$, where $a\in C[t,x]^r$ and $u_\infty\in C[t,x]$ is integer-linear. Let~$\tilde f = \frac{1}{u_\infty}aW$. Then $f = \frac{1}{\bar u}\tilde f$. Since $f$ is not proper, there exists an irreducible polynomial $p\in C[t,x]$ such that $p\mid \bar u$. Then $p$ is not integer-linear. Replacing $p$ by $\sigma_t^{i_0}(p)$ if necessary, where $i_0\in \set N $ is the maximal integer such that $\sigma_t^{i_0}(p)\mid \bar u$, we may assume that $p\mid \bar u$ and $\sigma_t^{i}(p) \nmid \bar u $ for all $i>0$.
	
	Let $e_t\in C[t,x]$ and $M_t\in C[t,x]^{r\times r}$ be such that $S_tW = \frac{1}{e_t}M_tW$. Then
	\begin{equation}\label{EQ: T tilde f}
		S_t^i\tilde f = \sigma_t^i\left(\frac{a}{u_\infty}\right)\sigma_t^{i-1}\left(\frac{1}{e_t}M_t\right) \cdots \sigma_t\left(\frac{1}{e_t}M_t\right) \frac{1}{e_t}M_t W
	\end{equation}
	is proper for all $i\geq 0$, because by Theorem~\ref{Thm: compatibility}.\eqref{it:compat2}, $e_t$ is integer-linear. We write $T=\sum_{i=0}^\rho c_iS_t^i\in C(t)[S_t]$, where $\rho\in\set N$, $c_i\in C(t)$ and $c_\rho\neq 0$. Then $c_i(t)$ is integer-linear. Since $S_t^if = S_t^i\frac{1}{\bar u} \tilde f=\sigma_t^i(\frac{1}{\bar u})S_t^i\tilde f$, by~\eqref{EQ: T tilde f} we obtain that
	\[Tf = \sum_{i=0}^\rho c_iS_t^if,\]
	where
	\begin{align*}
		S_t^if&=\sigma_t^i\left(\frac{1}{\bar u}\right)\sigma_t^i\left(\frac{a}{u_\infty}\right)\sigma_t^{i-1}\left(\frac{1}{e_t}M_t\right) \cdots \sigma_t\left(\frac{1}{e_t}M_t\right) \frac{1}{e_t}M_t W\\
		&=\sigma_t^i\left(\frac{1}{\bar u}\right) \underbrace{\sigma_t^i(\frac{1}{u_\infty})\prod_{j=0}^{i-1}\sigma_t^{j}\left(\frac{1}{e_t}\right)}_{\text{integer-linear}}\sigma_t^i(a) \underbrace{\prod_{j=0}^{i-1}\sigma_t^j(M_t)}_{=:M_i} W.
	\end{align*}
    Since $\gcd(a_1,\ldots, a_r, \bar uu_\infty) = 1$ and $p\mid \bar u$, it follows that $\gcd(\sigma_t^i(a_1), \ldots, \sigma_t^i(a_r), \sigma_t^i(p)) = 1$ for all $i\geq 0$. By Theorem~\ref{Thm: compatibility}.\eqref{it:compat2}, $\det(M_t)$ is integer-linear. So $\det(M_i)$ is integer-linear. Since $p$ is not integer-linear, $M_i\in K[x]^{r\times r}$ is an invertible matrix over $K[x]_{(\sigma_t^i(p))}$. By Lemma~\ref{Lem: gcd}, the greatest common divisor of all entries of the vector $\sigma_t^i(a)M_i$ and $\sigma_t^i(p)$ is $1$. So $\sigma_t^i(p)$ is a factor of the denominator of $S_t^i f$. On the other hand, for each $i'$ with $0\leq i'< i$, $\sigma_t^i(p)$ is not a factor of the denominator of $S_t^{i'}f$. Otherwise $\sigma_t^i(p)\mid \sigma_t^{i'}(\bar u)$, which implies $\sigma_t^{i-i'}(p)\mid \bar u$. This contradicts the choice of $p$. Therefore $\sigma_t^\rho(p)$ is a factor of the denominator of $Tf$. Since $p$ is not integer-linear, $\sigma_t^\rho(p)$ is a factor of the stem of $Tf$. 
    
    Suppose that the stem of $Tf$ is spread with respect to $x$. There exists $j_0\in \set Z\setminus\{0\}$ such that $\sigma_x^{j_0}\sigma_t^\rho(p)$ is a factor of the stem of $Tf$. Then $\sigma_x^{j_0}\sigma_t^\rho(p)$ is a factor of the denominator of $c_{i_0}S_t^{i_0}f$ for some $i_0$ with $0\leq i_0\leq \rho$. So $\sigma_x^{j_0}\sigma_t^{\rho}(p)\mid\sigma_t^{i_0}(\bar u)$. Thus both $p$ and $\sigma_t^{\rho-i_0}\sigma_x^{j_0}(p)$ divide $\bar u$. Since $j_0\neq 0$ and $\bar u$ is shift-free with respect to $x$, we have $\rho-i_0\neq 0$ and hence $\rho - i_0 > 0$. So for every irreducible factor $p$ of $\bar u$, there exist $i,j\in\set Z$, $i>0$ such that $\sigma_t^i\sigma_x^j(p)$ divides $\bar u$. Hence by Lemma~\ref{Lem: integer-linear thm6}, $\bar u$ is integer-linear. This leads to a contradiction.
\end{proof}

\begin{proof}[Proof of Theorem~\ref{Thm: main_shiftfree}] Suppose that $f=\frac{aW}{u}$, where $u$ is shift-free with respect to $x$, has a telescoper of type $(S_t;S_x)$. There exists a nonzero operator $T\in C(t)[S_t]$ such that $Tf$ is summable in $A$ with respect to $x$. By Theorem~\ref{Thm:spread_sum}, the stem of $Tf$ is spread with respect to $x$. By Theorem~\ref{Thm:spread_Tf}, this condition is not satisfied unless the stem of $f$ is $1$, i.e., $f$ is proper. Thus $u$ is integer-linear.
\end{proof}
\begin{proof}[Proof of Theorem~\ref{Thm: main} (continued)]
	Suppose that $f$ has a telescoper of type $(S_t;S_x)$. Let $W$ be a suitable basis of $A$ with respect to $x$. By the decomposition~\eqref{EQ: AP-reduction}, we get
	\[	f = \Delta_x(g)+ h\quad\text{and}\quad h=\frac{cW}{\tilde u},\]
where $g\in A$, $c\in K[x]^r$, $\tilde u\in K[x]$ is shift-free with respect to $x$. For any operator $T\in C(t)[S_t]$, we have $Tf = \Delta_x(Tg) + Th$. So $f$ has a telescoper of type $(S_t;S_x)$ if and only if $h$ has a telescoper of type $(S_t;S_x)$. By Theorem~\ref{Thm: main_shiftfree}, $\tilde u$ is integer-linear, i.e., $h$ is proper.
\end{proof}

\section*{Acknowledgement}
I would like to thank Manuel Kauers for helpful discussions, especially on the properties of properness in Section~\ref{sec: stem}. I also would like to thank Shaoshi Chen for his support.


\end{document}